\documentclass[11pt, letterpaper]{amsart}
\usepackage{graphicx, amssymb, color}
\usepackage{amsmath}
\usepackage{natbib}

\newtheorem{thm}{Theorem}[section]
\newtheorem{cor}[thm]{Corollary}
\newtheorem{lem}[thm]{Lemma}
\newtheorem{prop}[thm]{Proposition}
\theoremstyle{definition}
\newtheorem{defn}[thm]{Definition}

\newtheorem{ass}{Assumption}

\theoremstyle{remark}
\newtheorem{rem}[thm]{Remark}
\newtheorem{exa}[thm]{Example}

\numberwithin{equation}{section}
\numberwithin{figure}{section}

\newcommand{\set}[1]{\left\{#1\right\}}
\newcommand{\Real}{\mathbb R}

\newcommand{\Complex}{\mathbb C}
\newcommand{\F}{\mathcal{F}}

\newcommand{\prob}{\mathbb{P}}
\newcommand{\V}{\mathcal{V}}
\newcommand{\D}{\mathcal{D}}
\newcommand{\C}{\mathcal{C}}
\newcommand{\expec}{\mathbb{E}}
\newcommand{\I}{\mathcal{I}}
\newcommand{\fE}{\mathfrak{E}}
\newcommand{\fD}{\mathfrak{D}}

\newcommand{\basis}{(\Omega,  \, (\F_t)_{t \in \Real_+}, \, \prob)}
\newcommand{\indic}{\mathbb{I}}
\newcommand{\pare}[1]{\left(#1\right)}
\newcommand{\bra}[1]{\left[#1\right]}
\newcommand{\such}{\, | \, }

\title[On exponential moments of affine diffusions]{Long-Term and Blow-Up Behaviors of exponential moments in multi-dimensional Affine Diffusions}

\author[]{Rudra P. Jena}
\address[Rudra P. Jena]{Centre de Math\'ematiques Appliqu\'ees, Ecole Polytechnique,
91128 Palaiseau Cedex France}
\email{jena@cmap.polytechnique.fr}
\author[]{Kyoung-Kuk Kim}
\address[Kyoung-Kuk Kim]{Department of Industrial and Systems Engineering, Korea Advanced Institute of Science and Technology, Daejeon 305-701, South Korea}
\email{catenoid@kaist.ac.kr}
\author[]{Hao Xing}
\address[Hao Xing]{Department of Statistics, London School of Economics and Political Science, London WC2A 2AE, UK}
\email{h.xing@lse.ac.uk}
\thanks{The research of R. Jena was supported by the Chair Financial Risks of the Risk Foundation sponsored by Soci\'{e}t\'{e} G\'{e}n\'{e}rale, the Chair Derivatives of the Future sponsored by the F\'{e}d\'{e}ration Bancaire Francaise, and the Chair Finance and Sustainable Development sponsored by EDF and Calyon. The research of K. Kim was supported by Basic Science Research Program through the National Research Foundation of Korea (NRF) funded by the Ministry of Education, Science and Technology (2012-0003203). The research of H. Xing was supported by STICERD at London School of Economics.
}
\thanks{This work was initiated when three authors were visiting the Fields Institute for the thematic program on quantitative finance in 2010. The authors are grateful for the hospitality and support from the institute. We thank the two anonymous referees and the Associate Editor for their valuable comments, which helped us improve this paper.}

\date{May 15, 2012}

\begin{document}

\maketitle

\begin{abstract}
This paper considers multi-dimensional affine processes with continuous sample paths. By analyzing the Riccati system, which is associated with affine processes via the transform formula, we fully characterize the regions of exponents in which exponential moments of a given process do not explode at any time or explode at a given time. In these two cases, we also compute the long-term growth rate and the explosion rate for  exponential moments. These results provide a handle to study implied volatility asymptotics in models where log-returns of stock prices are described by affine processes whose exponential moments do not have an explicit formula.
\end{abstract}

\baselineskip18pt

\section{Introduction}

Since the introduction of the Black-Scholes model, many models have been developed to capture empirical features of financial asset prices. Among them, models proposed by \cite{CIR}, \cite{Heston}, \cite{Vasicek},  and many others have been widely used by market participants because of their analytical tractability in derivative pricing in addition to their ability to reflect observed market phenomena. Later, common features of these models were unified to introduce the notion of affine processes with the so called canonical state space. The general treatment of affine processes with this state space was conducted by \cite{DPS} and later extended by \cite{DFS}. Quite recently, studies on affine processes have been extended to more general state spaces; see, e.g., \cite{CuchieroF} and references therein.

A defining feature of affine processes is the logarithm of their Fourier transform is a linear function of the state. The regularity of affine processes, proved in \cite{Keller-Resselb} for affine processes on the canonical state space, links the aforementioned linear function to solutions to a system of (generalized) Riccati differential equations. This connection contributes to the analytical tractability of affine processes and enables us to express the values of derivative contracts, whose underlying is modeled by affine processes, via the Fourier inversion formula (\cite{LeeRW} and references therein). Moreover, this connection bridges affine processes and the theory of dynamical systems. Distributional properties of affine processes can be characterized by dynamical behaviors of solutions to the associated Riccati system. We refer the reader to \cite{Filipovic-Mayerhofer}, \cite{GK10}, \cite{Keller-Ressel}, and \cite{CuchieroF} for recent developments in this direction.

In this paper, we investigate long-term and blow-up behaviors of exponential moments of affine processes with the canonical state space $\Real_+^m \times \Real^n$. We treat general multivariate affine processes with continuous sample paths, so called \emph{affine diffusions}. This restriction of  affine processes to diffusions is imposed because its transform formula has been well understood in \cite{Filipovic-Mayerhofer}.
Currently the transform formula for affine processes with jumps is being studied; see \cite{SpreijV}. The generalization of our results to affine processes with jumps is left as future studies.
By focusing on affine diffusions, we are able to find sharp answers to the following two questions:
\begin{enumerate}
\item[\textit{Q1:}] Given an affine diffusion $X$, what are all possible vectors $u$ such that $\expec \exp(u^\top X_T)<\infty$ for any $T\geq 0$? For such a vector $u$, does the long-term growth rate $\lim_{T\rightarrow \infty}T^{-1} \log\{\expec \exp(u^\top X_T)\}$ exist?

\vspace{0.2cm}

\item[\textit{Q2:}] For a fixed $T>0$, what are all possible vectors $u$ such that $\expec \exp(u^\top X_S)<\infty$ for any $S<T$? For a vector $u$ such that $\expec \exp(u^\top X_S)$ is finite for all $S<T$ but infinity for $S=T$, does the blow-up rate $\lim_{S\uparrow T}(T-S) \log \{\expec \exp(u^\top X_S)\}$ exist?
\end{enumerate}
These questions are motivated by practical applications explained in the next paragraph, but they are also mathematically interesting. 
By focusing on a class of affine diffusions with some hierarchical structure between components (see Assumption \ref{ass: B^v upper-tri}),
we provide complete answers to \textit{Q1} and \textit{Q2}. (It should be noted that this class contains virtually all affine diffusions with the canonical state space in financial modeling.) The set of $u$ such that $\expec \exp(u^\top X_T)<\infty$ for any $T>0$ is characterized via the disjoint union of stable sets for equilibrium points of the Riccati system in Theorem \ref{thm: S-inf}. Moreover, the growth rate of exponential moments is identified in Corollary \ref{cor: exp moment X}. Working with a transformed Riccati system, similar answers to $\textit{Q2}$ are provided in Theorem \ref{thm:St-bdry} and Corollary \ref{cor2: exp moment X}. These results are extensions of \cite{GK10} and \cite{Keller-Ressel} to affine diffusions with arbitrary dimension. These findings not only help numerically identify sets of vectors in \textit{Q1} and \textit{Q2}, they also characterize large-time asymptotics and explosion phenomenon of exponential moments of multi-dimensional affine diffusions.

For the past several years, large-time asymptotics and explosion phenomena of stock price moments have attracted considerable attention because of their close connection to implied volatility asymptotics. By approximating long-term stock price moments, \cite{Lewis} derived an asymptotic formula for the implied volatility at large maturities in the fixed-strike regime under the Heston model. Recently, \cite{FordeJ2} obtained similar implied volatility asymptotics for the Heston model in a regime where the log-moneyness is proportional to the maturity. The first step in their analysis is to study the long-term behaviors of stock price moments (see Theorem 2.1 in \cite{FordeJ2}). On the other hand, it is well known that the explosion of certain moments of stock prices at fixed time $T$ is related to the implied volatilities at extreme strikes with option maturity $T$; see \cite{RLee} and \cite{BenaimF08} for extensions. For example, an upper bound on the asymptotic slope of implied volatilities of deep-out-of-money options is found to be a function of the critical exponent $p^*=\sup\{p \such \expec S^{p+1}_T <\infty\}$.
Such asymptotic values of implied volatilities are informational in extrapolating smile curves and in calibrating underlying models to market prices. More details about this practical usage can be found in, e.g., \cite{BenaimF08} and \cite{FordeJ2}. When the stock price log-return is modeled by an affine diffusion, results in this paper help to identify implied volatility asymptotics for large-time-to-maturity, deep-out-of-money or deep-in-the-money options; see Section \ref{subsec: applications} and three examples in Section \ref{sec:examples}.

The paper is structured as follows. In Section~\ref{sec: affine}, we review basic concepts of affine diffusions and their canonical representations. We present our main results in Section~\ref{sec:main}. Then, three multi-dimensional examples are presented to illustrate our findings in Section~\ref{sec:examples}. Analysis on the Riccati system and proofs of main results are developed in Sections~\ref{sec:moment} and \ref{sec:impvol}. Finally, Section~\ref{sec:conclusion} concludes.

Before we move on, let us introduce some notational conventions which will be used throughout the paper.
\begin{itemize}
\item For a vector $x$ in a Euclidean space, $|x|$ means its Euclidean norm regardless of dimension.
\item If $x, y$ are of the same dimension then $x \leq y$ if and only if $x_i \leq y_i$ for each component. And $x \cdot y$ represents the Euclidean inner product between $x$ and $y$.
\item For a vector in $\Real^{m+n}$ or a matrix in $\Real^{(m+n)\times (m+n)}$, we denote the first $m$ entries of the vector or $m\times m$ entries of the matrix by the superscript $\V$, and the last $n$ entries of the vector or $n\times n$ entries of the matrix by the superscript $\D$.
\item By $x^{(2)}_{\I}$, where $x \in \mathbb{R}^m$ and $\I \subset \{1, \ldots, m\}$, we mean a vector of which $i$-th entry is $x_i^2\indic_{i \in \I}$.
\item For matrices, $diag(x)$ for $x \in \mathbb{R}^m$ is the $m\times m$ diagonal matrix with $(x_1, \ldots, x_m)$ as its diagonal entries, and $diag_\I(x)$ with $\I \subseteq \set{1, \ldots, m}$ is the $m\times m$ diagonal matrix such that its $i$-th diagonal entry is $x_i \indic_{i \in \I}$. $I_k$ is the $k\times k$ identity matrix.
\item For a set $A$ in Euclidean space, $A^\circ$ is its interior and $A^c$ is its complement.
\end{itemize}

\section{Affine Diffusions on Canonical State Space}\label{sec: affine}
Let us recall affine diffusions and their canonical representation in this section. Given $b: \Real_+^m \times \Real^n \to \Real^{d}$ and $\sigma: \Real_+^m \times \Real^n \to \Real^{d\times d}$ for some nonnegative integers $m, n$ and $d=m+n$, we consider the following stochastic differential equation (SDE) on a probability space $\basis$:
$$
 dY_t = b(Y_t) \, dt+ \sigma(Y_t) \, dW_t, \quad Y_0=y,
$$
where $W$ is a $d$-dimensional standard Brownian motion and $y\in \Real_+^m \times \Real^n$. The above SDE admits a unique solution when $b$ and $\sigma$ are of affine type and satisfy \emph{admissible constraints} introduced below (see Theorem 8.1 in \cite{Filipovic-Mayerhofer}). The state space $\Real_+^m \times \Real^n$ is called the \emph{canonical state space}. In financial applications, the first $m$ components of $Y$, which usually model volatility processes, are called \emph{volatility state variables}; while the other $n$ components of $Y$ are called \emph{dependent state variables}. In this case, $Y^\V$ models the volatility variables and $Y^\D$ describes the dependent variables.

We say that $Y$ is an \emph{affine process} if there exist $\Complex$- and $\Complex^d$-valued functions $\phi$ and $\psi$ such that
\begin{equation}\label{eq: def affine}
 \mathbb{E}\left[ \exp\left(u^{\top} Y_T\right) | \mathcal{F}_t\right] = \exp\left( \phi(T-t, u) + \psi(T-t, u)^{\top} Y_t\right),
\end{equation}
for all $u\in i \Real^d$, $t\leq T$, and $y\in \Real_+^m \times \Real^n$. This specification implies that the diffusion matrix $a(y):= \sigma(y) \sigma(y)^\top$ and the drift $b(y)$ are both affine functions (see Theorem~2.2 in \cite{Filipovic-Mayerhofer}), i.e.,
$$
 a(y) = a + \sum_{i=1}^{d} y_i \alpha_i, \quad b(y) = b + \sum_{i=1}^{d} y_i \beta_i =: b + B y,
$$
for some $a, \alpha_i \in \Real^{d\times d}$ and column vectors $b, \beta_i\in \Real^d$ with $B:= (\beta_1 \cdots \beta_{d})\in \mathbb{R}^{d\times d}$. Moreover, regularity of affine processes proved by \cite{Keller-Resselb} ensures that $\phi$ and $\psi = (\psi_1, \cdots, \psi_d)$ satisfy the following system of Riccati differential equations:
\begin{eqnarray*}
\partial_t \phi(t, u) &=& \frac{1}{2} \psi(t, u)^\top a \,\psi(t,u) + b^\top \psi(t, u), \quad \phi(0, u) = 0, \\
\partial_t \psi_i(t, u) &=& \frac{1}{2} \psi(t, u)^\top \alpha_i\, \psi(t,u) + \beta_i^\top \psi(t, u),  \quad \psi(0, u) = u, \quad \text{ for } 1\leq i\leq d.
\end{eqnarray*}

To ensure that $Y$ is an affine process on the state space $\Real_+^m \times \Real^n$, we impose the following \emph{admissible constraints} on parameters $a, \alpha_i, b$, and $\beta_i$  (see Theorem~3.2 in \cite{Filipovic-Mayerhofer}):
\begin{enumerate}
\item[(i)] $a$, $\alpha_i$ are symmetric positive semi-definite, and $\alpha_{m+1} = \cdots = \alpha_{m+n} = 0$,
\item[(ii)] $a = \left( \begin{array}{cc}
                0 & 0 \\
                0 & a^\D \end{array}\right)$, $\alpha_i = \left( \begin{array}{cc}
                                                                    c_i\delta_{ii} & w_i \\
                                                                    w_i^\top & \alpha_i^\D \end{array}\right)$,
                where $c_i\in \Real$, $\delta_{ii}\in \Real^{m\times m}$ is the zero matrix except 1 for the $(i,i)$-th entry, and $w_i \in \mathbb{R}^{m\times n}$ has zero entries except the $i$-th row,
\item[(iii)] $b \in \Real_+^m \times \Real^n$, $B = \left( \begin{array}{cc}
                        B^\V & 0 \\
                        * & B^\D \end{array}\right)$, and $B^\V$ has nonnegative off-diagonal elements.
\end{enumerate}
Under these constraints, the transformation formula \eqref{eq: def affine} is extendable to real dimensions.
\begin{thm}[\cite{Filipovic-Mayerhofer}]\label{thm:FM} Suppose that $Y$ is an affine process with admissible parameters. Then, the transform formula \eqref{eq: def affine} holds true for $u \in \mathbb{R}^d$ as long as either side of the formula is finite. 
\end{thm}

In this paper, we focus on the following class of affine diffusions:
\begin{ass}\label{ass: B^v upper-tri}
 $B^{\V}$ is triangular (say, upper triangular) with strictly negative eigenvalues.
\end{ass}
The upper triangular shape of $B^{\V}$ imposes a hierarchical dependence structure between all volatility state variables. This hierarchical structure is commonly assumed in many financial models (see Section \ref{sec:examples} for several examples). In these models, different volatility state variables are usually used to model volatility processes on different time scales. Strictly negative eigenvalues imply that $Y^{\V}$ is mean-reverting, which is a natural property of volatility processes.

To facilitate our analysis on this class of affine diffusions, we consider their canonical representations (see Section 7 in \cite{Filipovic-Mayerhofer}). Given a linear transform $\Lambda: \Real_+^m \times \Real^n \rightarrow \Real_+^m \times \Real^n$, the process $X := \Lambda Y$ has the following dynamics:
$$
dX_t = (\hat b + \hat B X_t)\, dt + \hat \sigma(X_t)\, dW_t, \quad X_0 = \Lambda Y_0,
$$
where $\hat b = \Lambda b$, $\hat B = \Lambda B \Lambda^{-1}$, and $\hat \sigma(x) = \Lambda \sigma(\Lambda^{-1}x)$.
The transformed diffusion matrix is
$$
\hat a(x) = \hat\sigma(x)\hat\sigma(x)^\top =  \Lambda a\Lambda^\top + \sum_{i=1}^{d} (\Lambda^{-1}x)_i \Lambda\alpha_i\Lambda^\top =: \hat a + \sum_1^{m+n} x_i \hat \alpha_i.
$$
Actually, one can find a special $\Lambda\in \Real^{d\times d}$ 
 with diagonal $\Lambda^\V$, such that the diffusion matrix of $X$ has the following canonical form
 $$
 \hat a(x) = \left( \begin{array}{cc}
                    {diag}_{\mathcal{I}}(x) & 0 \\
                    0 & \pi_0 + \sum_1^m x_i \pi_i \end{array} \right),
$$
where $\I$ is a subset of $\set{1, \cdots, m}$,  and $\pi_i$, $0 \leq i \leq m$, are some symmetric positive semi-definite matrices in $\mathbb{R}^{n\times n}$. Moreover, the parameters of $X$ are admissible. (See Lemma~7.1 in \cite{Filipovic-Mayerhofer}.) Note that $\hat{B}^{\V}$ is still upper triangular since $\Lambda^\V$ is diagonal. Moreover, to exclude trivial cases where $Y^\V$ have deterministic dynamics, we assume that
$\I$ is non-empty.

In the canonical version, the Riccati system reads
\begin{eqnarray*}
\partial_t \phi(t, u) &=& \frac{1}{2} {\psi^{\D}(t, u)}^\top \pi_0 \,\psi^{\D}(t,u) + \hat b^\top \psi(t, u), \\
\partial_t \psi_i(t, u) &=& \frac{1}{2} \psi_i(t, u)^2 \,\indic_{i \in \mathcal{I}} + \frac{1}{2}\psi^{\D}(t, u)^\top \pi_i \, \psi^{\D}(t, u) + \hat \beta_i^\top \psi(t, u), \quad 1 \leq i \leq m,\\
\partial_t \psi_i(t, u) &=& \hat \beta_i^\top \psi(t, u), \quad m+1\leq i \leq m+n, 
\end{eqnarray*}
with initial conditions $\phi(0, u)=0$ and $\psi(0,u) = u$. Note that the first equation is easy to solve once we know $\psi$, hence we focus on equations for $\psi$ and write them succinctly as follows:
\begin{equation}\label{eq:riccati}
\begin{split}
 & \dot y = f(y,z), \hspace{2cm} y(0) = v,\\
 & \dot z = A^\D  z, \hspace{2.3cm} z(0) = w.
\end{split}
\end{equation}
Here $u=(v,w)$ with $v \in \Real^m$ and $w \in \Real^n$, $f=(f_1, \cdots, f_m)^\top$ with
\[f_i(y, z):= \frac 12 y_i^2 \, \indic_{i\in \I} + \sum_{k=1}^i A_{ik} y_k + g_i(z), \quad \text{ in which}\]
\[
  g(z) = \pare{\begin{array}{c} g_1(z) \\ \vdots \\ g_m(z) \end{array}} := \frac{1}{2} \left( \begin{array}{c}
                                                        {z}^\top \pi_1 \,z \\
                                                        \vdots \\
                                                        {z}^\top \pi_m \, z \end{array}\right) + A^{\C}\, z
  \quad \text{ and } \quad A = \left( \begin{array}{cc}
            A^{\V} & A^{\C} \\
            0 & A^{\D} \end{array} \right) := \left( \begin{array}{cc}
                                                            \hat B^{\V} & 0 \\
                                                            * & \hat B^{\D} \end{array}\right)^\top.
 \]

Assumption \ref{ass: B^v upper-tri} implies that $A^{\V}$ is a lower triangular matrix. Hence $(y_1, \cdots, y_i)$ in \eqref{eq:riccati} is an autonomous system for each $i\in \set{1, \cdots, m}$ when $w \in \operatorname{Ker} A^\D$. Moreover, $A^{\V}$ has strictly negative eigenvalues with nonnegative off-diagonal elements, whence $- A^{\V}$ is a nonsingular M-matrix (see Definition~\ref{def: M-matrix}). Now the transform formula reads
\begin{equation}\label{eq: tran formula}
 \expec\bra{\exp(u^\top X_T)| \F_t} = \exp\pare{I(T-t) + y(T-t)\cdot X^\V_t + z(T-t) \cdot X^\D_t},
\end{equation}
where $I(\cdot):= (1/2) \int_0^\cdot z(s)^\top \pi_0 z(s) \, ds + \int_0^\cdot \hat{b}^\V \cdot y(s)\, ds+ \int_0^\cdot \hat{b}^\D \cdot z(s)\, ds$ and $(y, z)$ solves \eqref{eq:riccati}.

In financial applications, the discounted stock price, say $S$, is usually modeled by an affine process $X$ via $S_\cdot = \exp(\theta^\top X_\cdot)$ for some $\theta \in \Real^d$. Then, $S$ being a martingale under (a risk neutral measure) $\prob$ implies that $\theta^\D \in \operatorname{Ker} A^\D$ (see \eqref{eq: mart}). Therefore, in this paper, we always choose the initial condition for the second equation in \eqref{eq:riccati} to be $z(0) = w \in \operatorname{Ker} A^{\D}$. Hence, $z(t)= w$ for any $t\geq 0$, and the first equation in \eqref{eq:riccati} reads
\begin{equation}\label{eq:riccati-V}
 \dot y = f(y, w), \quad y(0)= v. \tag{Ric-V}
\end{equation}
We call $v\in \Real^m$ an \emph{equilibrium point} of \eqref{eq:riccati-V} if $f(v, w)=0$.

\begin{rem}
  It is of potential mathematical interest to consider models without Assumption \ref{ass: B^v upper-tri}. However, in such cases, even identifying all equilibrium points of \eqref{eq:riccati-V} becomes a nontrivial task as we need to solve a system of
  coupled algebraic equations.
  Still, there is one case where some of the results in this paper can be obtained to some extent, and this is when $A^\D$ is invertible. We refer the reader to \cite{Kim} for details.
\end{rem}

\section{Main Results}\label{sec:main}
In this section, we present our main results whose proofs are deferred to Sections~\ref{sec:moment} and \ref{sec:impvol}. In Section \ref{subsec:longterm}, we look for all $u\in \Real^m \times \operatorname{Ker} A^\D$ such that $\expec[\exp(u^\top X_T)]$ is finite for all $T\geq 0$. In Section \ref{subsec:blow-up}, we characterize all $u\in \Real^m \times \operatorname{Ker} A^\D$ such that $\expec[\exp(u^\top X_S)]$ is finite for all $S$ before a given $T$. Applications of these characterizations to financial modelings are given in Section \ref{subsec: applications}.

\subsection{Long-term behaviors}\label{subsec:longterm}
Our first result identifies $u\in \Real^m \times \operatorname{Ker} A^\D$ such that $\expec[\exp(u^\top X_T)]$ is finite for all $T\geq 0$. Thanks to Theorem \ref{thm:FM}, the problem is equivalent to finding every initial condition $v \in \Real^m$ such that the solution $y$ to \eqref{eq:riccati-V} does not blow up in finite time. To this end, let us first classify equilibrium points of \eqref{eq:riccati-V} into several different types, each of which tells us about qualitative behaviors of solutions in a neighborhood of an equilibrium point. See \cite{Chiang88} or \cite{Perko} for more backgrounds.

\begin{defn}\label{def: stable eq point}
 An equilibrium point $\nu\in \Real^m$ of \eqref{eq:riccati-V} is \emph{stable} if for each $\epsilon > 0$ there exists $\delta>0$ such that $\|y(t)-\nu\| < \epsilon$ for all $t >0$ whenever $\|y(0)-\nu\|< \delta$. It is  \emph{asymptotically stable} if it is stable and $\lim_{t\uparrow \infty} y(t) =\nu$. Otherwise, $\nu$ is \emph{unstable}. Also, if all eigenvalues of the Jacobian $Df(\nu) = A^\V + diag_\I(\nu)$ of $f$ at $\nu$ have nonzero real parts, then $\nu$ is \emph{hyperbolic}.
\end{defn}

The following result identifies all asymptotically stable equilibrium points for \eqref{eq:riccati-V}.
\begin{lem}\label{lem: eq point summary}
 There exists a nonempty closed convex set $\fD \subset \operatorname{Ker} A^\D$ with the following properties. First of all, for each $w\in \fD$, there are at most $2^{|\I|}$ equilibrium points for \eqref{eq:riccati-V}. Second, for each $w\in \fD^\circ$, $\eta(w) = (\eta_1(w), \cdots, \eta_m(w))$, where
 \begin{equation}\label{eq: def eta}
  \eta_i(w) := \left\{\begin{array}{cc}
    -A_{ii} - \sqrt{A_{ii}^2 - 2\left( \sum_{k=1}^{i-1}A_{ik}\eta_k(w) + g_i(w)\right)} & {\rm if}\; i \in \mathcal{I}\\
    -A_{ii}^{-1}\left( \sum_{k=1}^{i-1}A_{ik}\eta_k(w) + g_i(w)\right) & {\rm if}\;  i \in \set{1, \cdots, m}\setminus \I  \end{array}\right.,
 \end{equation}
 is hyperbolic and it is the unique asymptotically stable equilibrium point. All other equilibrium points are unstable, while at least one of them is hyperbolic. Lastly, there is no equilibrium point when $w\in \operatorname{Ker} A^\D \cap \fD^c$.
\end{lem}

The construction of $\fD$ is explicit (see \eqref{def: D} below). Moreover, $\eta(w)$ can be determined sequentially from $i=1$ to $i=m$ since $A^{\V}$ is lower triangular with strictly negative diagonal entries. Now in order to connect the long-term behavior of solution trajectories to equilibrium points, we introduce  the following notion.
\begin{defn}\label{def: stable set}
 Given an equilibrium point $\nu$ of \eqref{eq:riccati-V}, its \emph{stable set} is
\[
  W^s_{\nu}(w) := \set{v\in \Real^m \such \lim_{t\uparrow \infty} y(t) = \nu  \text{ where } y(t) \text{ solves }\eqref{eq:riccati-V}}.
 \]
 When $\nu = \eta(w)$, we write $W^s_{\nu}(w)$ as $\mathcal{S}(w)$ and call it the \emph{stable region} of \eqref{eq:riccati-V}.
\end{defn}
Another related object is the set of initial conditions for \eqref{eq:riccati} such that its solution trajectory does not explode in finite time:
\[
 \mathcal{S}_\infty := \set{u = (v,w) \in \Real^m\times \operatorname{Ker} A^\D \such |y(t)|<\infty \text{ for all } t\in\Real_+, \text{ where } y \text{ solves } \eqref{eq:riccati-V}}.
\]
For each $w$, $\mathcal{S}_\infty(w)$ is the section of $\mathcal{S}_\infty$, i.e. $\mathcal{S}_\infty(w):= \set{v\in \Real^m \such (v,w)\in \mathcal{S}_\infty}$.
We are now ready to state our first main result, which provides a decomposition of $\mathcal{S}_\infty$.
The interior of $\mathcal{S}_\infty$ is the disjoint union of stable regions $\mathcal{S}(w)$ for all $w\in \fD^\circ$, the boundary of $\mathcal{S}_\infty$ consists of two components: 1. disjoint union of all stable sets of nonstable equilibria $\nu$ for each $w\in \fD^\circ$, 2. disjoint union of $\mathcal{S}_\infty(w)$ for each $w\in \partial\fD$. In all of our statements, the topology is the relative Euclidean topology of $\Real^m \times \operatorname{Ker} A^D$.

\begin{thm}\label{thm: S-inf}
The interior and the boundary of $\mathcal{S}_\infty$ have the following decompositions in $\mathbb{R}^m\times \operatorname{Ker} A^\D$:
\begin{enumerate}
\item[i)]
$
\mathcal{S}^\circ_\infty = \bigcup_{w\in \fD^\circ} \mathcal{S}(w) \times \set{w}.
$
\item[ii)]
$
\partial \mathcal{S}_\infty = \bigg( \bigcup_{w \in \fD^\circ} \bigcup_{\nu\neq\eta(w)}W^s_\nu(w) \times \set{w} \bigg) \bigcup \bigg( \bigcup_{w \in \partial \fD} \mathcal{S}_\infty(w)\times \set{w}\bigg),
$
where $\nu$ is chosen from equilibrium points of \eqref{eq:riccati-V}. Moreover, for each $w\in \partial \fD$, there exists a nonempty set $\mathcal{M} \subset \{1, \cdots, m\}$ such that the set $\{v_{\mathcal{M}} \such v \in \mathcal{S}_\infty(w)\}$ is the stable set of the following system
\[
 \dot y_i = \frac12 y_i^2 \indic_{i\in \I} + \sum_{k\in \mathcal{M}} A_{ik} y_k + g_i(w), \quad \text{ for } i\in \mathcal{M},
\]
which admits a unique equilibrium point. Here $v_{\mathcal{M}} := (v_{i_1}, \ldots, v_{i_k})$ if
$\mathcal{M} = \set{i_1, \ldots, i_k} \subset \set{1, \ldots, m}$.
\end{enumerate}
\end{thm}

In some special cases, the description of $\partial \mathcal{S}_\infty$ becomes succinct. A hyperbolic equilibrium point $\nu$ of \eqref{eq:riccati-V} is of \emph{type} $k$ if it admits $k$ eigenvalues with positive real parts in its Jacobian matrix. It is a standard result in dynamical systems theory that the stable set of an equilibrium point $\nu \in \Real^m$ of type $k$ is a smooth manifold of dimension $m-k$. If the system of interest has hyperbolic equilibrium points only, then the description of the $(m-1)$-dimensional object $\partial \mathcal{S}_\infty(w)$ for $w \in \fD^\circ$ does not need the stable sets for equilibrium points of type $k > 1$, because these stable sets have $(m-1)$-dimensional Lesbegue measure zero.

\begin{cor}\label{cor: good case} Suppose that $A^\D$ is invertible and every equilibrium point of \eqref{eq:riccati-V} is hyperbolic. Then, $\partial \mathcal{S}_\infty$ is given by $\bigcup_\nu W^s_\nu(0) \times \set{0}$ except a set of $(m-1)$-dimensional Lesbegue measure zero. Here, $\nu$ is chosen from hyperbolic equilibria of type 1 and $W^s_\nu(0)$ is a smooth manifold of dimension $m-1$.
\end{cor}

Going back to the  affine diffusion $X$, the characterization of $\mathcal{S}_\infty$, together with Theorem \ref{thm:FM}, helps to identify the long run behavior of its exponential moments.
\begin{cor}\label{cor: exp moment X}
The following statements are equivalent:
\begin{enumerate}
\item[i)] $\expec\left[\exp\left(u^\top X_T\right)\right]$ is finite for all $T\geq 0$ and $X_0\in \Real_+^m \times \Real^n$.
\item[ii)] $u\in \mathcal{S}_\infty$.
\end{enumerate}
Moreover, when either of these statements holds true,
\begin{equation}\label{eq: exp moment growth rate}
  \lim_{T\to \infty} \frac{1}{T} \log \expec\bra{\exp(u^\top X_T)} = \frac12 (u^\D)^\top \pi_0 u^\D + \hat{b}^\D \cdot u^\D + \left\{\begin{array}{ll}  \hat{b}^\V \cdot \eta(u^\D), & \text{ if } u \in \mathcal{S}_\infty^\circ;\\
  \hat{b}^\V \cdot \nu, & \text{ if } u \in  \partial \mathcal{S}_\infty, \end{array}\right.
\end{equation}
where $\nu$ is some unstable equilibrium point of \eqref{eq:riccati-V}.
\end{cor}

These findings connect to existing results in two ways. First, it generalizes characterizations in Proposition 5.2 of \cite{GK10} and Theorem 3.4 in \cite{Keller-Ressel} to multi-dimensional affine diffusions. In these two papers, similar characterizations on exponential moments are obtained in the canonical affine term structure model of \cite{DaiS} and 2-dimensional affine stochastic volatility models, respectively. Second, following the same arguments in Theorem 3.4 of \cite{Keller-Ressel}, Corollary \ref{cor: exp moment X} shows a certain similarity between large time moment generating functions of $X$ and a L\'{e}vy process whose characteristic exponent is given by the right hand side of \eqref{eq: exp moment growth rate}.

\subsection{Blow-up behaviors}\label{subsec:blow-up}
Given $T> 0$, our second result identifies $u\in \Real^m \times \operatorname{Ker} A^\D$ such that $\expec\left[\exp\left(u^\top X_S\right)\right]$ is finite for any $S < T$. To this end, let us first define the \emph{blow-up time} for solutions to \eqref{eq:riccati-V}.

\begin{defn}\label{def: blow-up time}
For the initial condition $u \in \Real^m \times \operatorname{Ker} A^\D$, the blow-up time $T^*(u)$ of a solution $y$ to \eqref{eq:riccati-V} is the first time $t^*$ such that $\lim_{t \rightarrow t^*}|y(t)| = \infty$.
\end{defn}

The following result, whose proof is deferred to Section~\ref{sec:impvol}, ensures the continuity of $u\mapsto T^*(u)$.
\begin{lem}\label{lem:cont blowup time}
 The blow-up time $T^*(\cdot)$ is continuous on the set $\mathcal{P} := \set{u\such T^*(u) < \infty}$.
\end{lem}

Similar to $\mathcal{S}_\infty$ in the last subsection, we define the set of initial conditions such that solutions to \eqref{eq:riccati} do not blow up before $T$:
\[
 \mathcal{S}_T := \set{u=(v,w) \in \Real^m\times \operatorname{Ker} A^\D \such |y(s)|<\infty, \forall s< T, \text{ where } y \text{ solves } \eqref{eq:riccati-V}}.
\]
We also define $\mathcal{S}_T(w)$ as a section of $\mathcal{S}_T$ for fixed $w \in \operatorname{Ker} A^\D$. It is apparent that $\mathcal{S}_T(w) = \set{v \such T^*(u) \geq T \text{ where } u=(v,w)}$ and that $\set{\mathcal{S}_T(w)}_{T\geq 0}$ is a decreasing sequence of sets, yielding $\bigcap_{T> 0} \mathcal{S}_T (w) = \bigcap_{T > 0} \set{v \such T^*(u) \geq T} = \mathcal{S}_{\infty}(w)$. This observation and Lemma \ref{lem:cont blowup time} combined indicates that both $\mathcal{S}_\infty(w)$ and $\mathcal{S}_T(w)$ are  closed sets in $\Real^m$. Moreover, \cite{Filipovic-Mayerhofer} showed that  $\mathcal{S}_T(w)$ is a convex neighborhood of the origin in $\mathbb{R}^m$. Similar conclusions hold for $\mathcal{S}_T$ and $\mathcal{S}_\infty$ in $\Real^m\times \operatorname{Ker} A^{\D}$ as well. On the other hand, it is not difficult to see from the definition of $\mathcal{S}_T$ that $\mathcal{S}_T^\circ=\set{u \such T^*(u) > T}$, hence $\partial \mathcal{S}_T = \set{u \such T^*(u) = T}$.

In what follows, we will characterize $\mathcal{S}_T$ and its boundary via the stability analysis of a transformed version of \eqref{eq:riccati-V}. Before we proceed, observe that $\mathcal{S}_T = \left\{u \such \expec\left[\exp\left(u^\top X_S\right)\right] <\infty,  \forall S< T\right\}$ from Theorem \ref{thm:FM}.  Hence, the study of $\mathcal{S}_T$ and its boundary is equivalent to investigating the blow-up behaviors of exponential moments of $X$.

Let us consider the following change of variables, inspired by \cite{Goriely}:
$$
x_i\left(s\right) := e^{-s}y_i\left(T(1 - e^{-s})\right), \quad i=1, \ldots, m, \quad x_{m+1}(s) := e^{-s}.
$$
Observe that if $x_i$ blows up at some $s^*>0$, then $y_i$ blows up at $T(1-e^{-s^*})<T$. Therefore if $y_i$ explodes at $T$, then $x_i$ does not explode in finite time. In addition, if $y_i$ explodes after $T$, then $\lim_{s\uparrow \infty} x(s) =0$. Given $y(0)= v \in \mathbb{R}^m$, one checks that $x$ satisfies the system of ODEs:
\begin{equation}\label{eq:quadODE}
\dot x_i = \frac{T}{2}x_i^2\indic_{i \in \mathcal{I}} - x_i + T\sum_{k=1}^i A_{ik}x_kx_{m+1} +T x_{m+1}^2g(w), \quad i = 1, \ldots, m,
\end{equation}
with $\dot x_{m+1} = - x_{m+1}$ and the initial condition $x(0) = (v, 1)$.
We introduced the auxiliary component $x_{m+1}$ to ensure the system \eqref{eq:quadODE} is autonomous.
Then, we observe that the equilibrium points of (\ref{eq:quadODE}) are given by $\nu'$ with $\nu'_i = 0$ or $2/T$, for $i \in \mathcal{I}$, and zero for all other indices. Also, every equilibrium point is hyperbolic, since the Jacobian at each equilibrium point has eigenvalues 1 or $-1$. Furthermore, the origin is the unique asymptotically stable equilibrium point of the system. For each equilibrium point $\nu'$ for \eqref{eq:quadODE}, let us denote its first $m$ components by $\nu$ and define the following stable set for $\nu$:
\[
 W^s_\nu(w,T) := \left\{v\in \Real^m \such \lim_{t\uparrow \infty} x(t) = \nu' \text{ where } x(t) \text{ solves } \eqref{eq:quadODE}\right\}.
\]
We are now ready to state our second main result, which characterizes the interior and the boundary of $\mathcal{S}_T$ as the disjoint unions of stable sets of equilibrium points for \eqref{eq:quadODE}.

\begin{thm}\label{thm:St-bdry} For each $T > 0$, the interior and the boundary of $\mathcal{S}_T$ have the following decompositions in $\Real^m\times \operatorname{Ker} A^{\D}$:
\[
\mathcal{S}_T^\circ = \bigcup_{w\in \operatorname{Ker} A^\D} W^s_0(w,T)\times\set{w} \quad {\text and}\quad \partial \mathcal{S}_T =  \bigcup_{w \in \operatorname{Ker} A^\D} \pare{\bigcup_{\nu \neq 0} W^s_{\nu}(w, T) \times \set{w}},
\]
where $\nu$ is chosen from the first $m$ components of equilibrium points of \eqref{eq:quadODE}.
\end{thm}

In the same spirit of Corollary~\ref{cor: exp moment X}, the blow-up behavior of exponential moments is identified as follows.
\begin{cor}\label{cor2: exp moment X}
For each $T > 0$, the following statements are equivalent:
\begin{enumerate}
\item[i)] $\expec\left[\exp\left(u^\top X_S\right)\right]$ is finite for all $S< T$ and $X_0\in \Real_+^m \times \Real^n$.
\item[ii)] $u\in \mathcal{S}_T$.
\end{enumerate}
If $u \in \partial \mathcal{S}_T$, then
\begin{equation}\label{eq: blow-up rate}
  \lim_{S \uparrow T} (T-S) \log \expec\bra{\exp(u^\top X_S)} = T \nu \cdot X^\V_0,
\end{equation}
where $\nu$ is the first $m$ components of some unstable equilibrium point of \eqref{eq:quadODE}.
\end{cor}

\subsection{Financial applications}\label{subsec: applications}
Affine processes have been widely used to model the stock price dynamics because of their analytical tractability in derivative pricing. In many models, the {\it discounted} stock price is represented by $S_\cdot = \exp(\theta^\top X_\cdot)$ for some $\theta \in \Real^d$. Let us assume that $S$ is a martingale under a risk neutral measure $\prob$. Then, this assumption is equivalent to the following conditions:
\begin{equation}\label{eq: mart}
 \theta \text{ is an equilibrium point of } \eqref{eq:riccati} \quad \text{ and } \quad 1/2 \,(\theta^\D)^\top \pi_0 \theta^\D + \hat{b}^\top \theta =0.
\end{equation}
In particular, $\theta^\D \in \operatorname{Ker} A^\D$. To prove \eqref{eq: mart}, we have from \eqref{eq: tran formula} that
$$
 1=\frac{\expec[S_T \such \F_t]}{S_t} = \exp\Big[I(T-t) + (y(T-t) - \theta^{\V}) \cdot X^{\V}_t + (z(T-t) - \theta^{\D}) \cdot X_t^{\D}\Big], \quad \forall \; t \leq T \text{ and } X_t,
$$
if and only if $I(s)=0$ and $(y(s), z(s))=(\theta^\V, \theta^\D)$ for any $s\in \Real_+$. Hence \eqref{eq: mart} is confirmed.

Now, the long run behavior of stock prices in this model follows from Corollary \ref{cor: exp moment X} directly.

\begin{prop}\label{thm: long term}
 For $\lambda \in \Real$, the following statements are equivalent:
 \begin{enumerate}
  \item[i)] $\mathbb{E}[S_T^\lambda]$ is finite for any $T\geq 0$ and $X_0 \in \Real_+^m\times \Real^n$.
  \item[ii)] $\lambda\theta \in \mathcal{S}_\infty$.
 \end{enumerate}
 When either of the above statements holds true, the asymptotical growth rate of the stock price moment is given by
 \begin{equation}\label{eq: moment growth rate}
  \lim_{T\to \infty} \frac{1}{T} \log \expec[S_T^\lambda] = \frac12 \lambda^2 (\theta^\D)^\top \pi_0 \theta^\D + \lambda \hat{b}^\D \cdot \theta^\D + \left\{\begin{array}{ll}  \hat{b}^\V \cdot \eta(\lambda \theta^\D), & \text{ if } \lambda \theta \in \mathcal{S}_\infty^\circ\\
  \hat{b}^\V \cdot \nu, & \text{ if } \lambda \theta\in  \partial \mathcal{S}_\infty \end{array}\right.,
 \end{equation}
 where $\nu$ is some unstable equilibrium point of \eqref{eq:riccati-V}.
\end{prop}

The characterization above has implications on prices of securities with super-linear payoffs. \cite{AndersenP} discuss possible unbounded prices of securities under two-factor affine or non-affine stochastic volatility models. Moreover, \eqref{eq: moment growth rate} can be used for the large-time-to-maturity implied volatilities for European options in multi-dimensional affine models. These asymptotic formulae facilitate calibrating models to implied volatility surfaces in practice and have been obtained in \cite{Lewis} and \cite{FordeJ2} for one volatility factor models. We provide several examples of multi-dimensional volatility factor models in Section \ref{sec:examples}.

\begin{rem}
  In \cite{FordeJ2} and \cite{Keller-Ressel}, a parametric constraint was imposed when the long-term growth rate was calculated. As we will see in Section \ref{sec:examples}, this parameter constraint is equivalent to $\theta$ being a stable equilibrium point. But, \eqref{eq: moment growth rate} still holds even when $\theta$ is unstable.
\end{rem}

The characterization of blow-up regions in Theorem \ref{thm:St-bdry} ties closely to the implied volatility asymptotics at extreme strikes for European options with fixed maturities. Let us denote $\sigma^2(x, T)$ the implied volatility for a European option with strike $K$, maturity $T$ and the log-moneyness $x=\log(K/S_0)$. \cite{RLee} proved that
\begin{equation}\label{eq: moment formula}
\limsup_{x \rightarrow \infty} \frac{\sigma^2(x, T)}{|x|/T} = \varsigma(p^*), \quad \limsup_{x \rightarrow -\infty} \frac{\sigma^2(x, T)}{|x|/T} = \varsigma(q^*),
\end{equation}
where $p^* = \sup\{p\geq 0 \such \mathbb{E}[S_T^{p+1}] < \infty\}$, $q^* = \sup\{q\geq 0 \such \mathbb{E}[S_T^{-q}] < \infty\}$, and $\varsigma(x) = 2 - 4(\sqrt{x^2+x}-x)$. Here $p^*$ and $q^*$ are called \emph{critical exponents}. This result was extended later by \cite{BenaimF08}, where the limit superiors in \eqref{eq: moment formula} are replaced by limits.
These asymptotic values of implied volatilities at extreme strikes have been found to be useful for extrapolation of smile curves (see  \cite{BenaimF08}). It is then vital to calculate critical exponents of underlying models in order to apply aforementioned connections to implied volatility asymptotics.

In the model where the logarithm of the discounted stock price is $\theta^\top X$, critical exponents can be identified by looking at $\partial \mathcal{S}_T$.
Before we proceed, let us first re-define critical exponents because they depend on the initial condition $X_0$ in our multi-dimensional setting. For example, consider a case where each component of $X_0$ is zero as long as the corresponding component of $y$ blows up at $T$. Then, the blow-up time of exponential moments would not be equal to that of $y$. Thus, we set $p^*$ as follows:
\begin{eqnarray*}
p^* &:=& \sup\set{ p \such \mathbb{E}[S_T^{p+1}] < \infty, \forall\, X_0 }= \sup\set{ p \such \mathbb{E}\bra{\exp((p+1)\theta^\top X_T)} < \infty, \forall\, X_0}\\
&=& \sup\set{ p \such |y(T)| < \infty \textrm{ with } (y(0), z(0)) = (p+1)\theta}= \sup\set{ p \such T^*((p+1)\theta) > T},
\end{eqnarray*}
where the third equality follows from \eqref{eq: tran formula} and Theorem~\ref{thm:FM}, the fourth equality holds since $T^*(\cdot \theta)$ is nonincreasing (see Lemma~\ref{l:derivT} below). We also note that $p^* \geq 0$  because the martingale property of $S$ implies that $\mathbb{E}[\exp(\theta^\top X_T)] = \exp(\theta^\top X_0) < \infty$ for any $T$. Similarly, we re-define $q^*: = \sup\set{ q \such T^*(-q\theta) > T}$. Now it follows from the continuity of $T^*$ (see Lemma~\ref{lem:cont blowup time}) that $T^*((p^*+1)\theta) = T^*(-q^*\theta) = T$. Hence, the intersections of a line passing through the origin and $\theta$ with $\partial \mathcal{S}_T$ yield critical exponents.

\section{Examples}\label{sec:examples}

\subsection{Heston model}
Let us start with the Heston model which is a prominent example of two-dimensional affine stochastic volatility model. The dynamics is determined by the SDE:
\[
dY_t =\left( \left( \begin{array}{c}
            \kappa\varphi \\
            0 \end{array}\right) -
          \left(  \begin{array}{cc}
            \kappa & 0 \\
            1/2 & 0 \end{array} \right) Y_t  \right) dt +
            \sqrt{Y^1_t}\left( \begin{array}{cc}
                    \sigma & 0 \\
                    \rho & \sqrt{1 - \rho^2} \end{array}\right) dW_t,
\]
where $W_t$ is a standard two-dimensional Brownian motion. The discounted stock price is modeled by $S_t = \exp(Y^\D_t)$, the variance process is described by $Y^\V$, and the diffusion matrix is
$
a(y) = y_1\left( \begin{array}{cc}
                                        \sigma^2 & \sigma\rho \\
                                        \sigma\rho & 1\end{array}\right)
$.
Choose $\Lambda = \left( \begin{array}{cc}
                            1/\sigma^2 & 0 \\
                            -\rho/\sigma & 1 \end{array}\right)$.
The canonical version $X=\Lambda Y$ has dynamics
$dX_t = (\hat b +\hat B X_t)dt + \hat\sigma(X_t)dW_t$ with
\[
 \hat b = \left( \begin{array}{c}
            \kappa\varphi/\sigma^2\\
            - \kappa\varphi\rho/\sigma \end{array}\right), \quad
 \hat B = \left( \begin{array}{cc}
                -\kappa & 0 \\
                \kappa\rho\sigma - \sigma^2/2 & 0 \end{array}\right), \quad
 \hat a(y) = \left( \begin{array}{cc}
            1 & 0 \\
            0 & \sigma^2(1-\rho^2) \end{array}
            \right) y_1,
 \]
 and the initial condition is given by $X_0=(
            V_0/\sigma^2 ,
            -\rho V_0/\sigma + \log S_0)$ where $S_0, V_0$ are the initial stock and variance levels.
 Moreover, $\log (S_T) = (\rho\sigma, 1)\cdot X_T$, thus $\theta = (\rho\sigma, 1)^\top$.

For any initial condition $(y(0), z(0)) = (v, w)$ of \eqref{eq:riccati}, $z(t) = w$ and $y(t)$ solves $\dot y = (1/2)y^2 - \kappa y + g(w)$ where $g(w) = \sigma^2 (1-\rho^2)w^2/2 + (\kappa \rho \sigma -\sigma^2/2) w$. It is clear that $\operatorname{Ker} A^\D = \Real$.
In addition, $\fD = \set{w \such 2g(w)\leq \kappa^2}$ and $\fD^\circ = \set{w \such 2g(w) < \kappa^2}$, following the definitions in Section~\ref{subsec: stable eq}. For each $w \in \fD^\circ$, the equation for $y$ admits two equilibrium points:  $L(w) = \kappa - \sqrt{\kappa^2-2g(w)}$ and  $U(w)=\kappa + \sqrt{\kappa^2-2g(w)}$. The former is asymptotically stable and hyperbolic, and the latter is unstable and hyperbolic as the Jacobian $Df(w) = - \sqrt{\kappa^2-2g(w)}<0$ ($\sqrt{\kappa^2-2g(w)}>0$) at $L(w)$ ($U(w)$ resp.). Therefore, the stable region is $\mathcal{S}(w) = (-\infty, U(w))$ hence $\mathcal{S}_\infty^\circ = \set{(v,w) \such w \in \fD^\circ, \; v < U(w)}$. The boundary $\partial \mathcal{S}_\infty$ is readily obtained as well. These sets are illustrated in the left panel in Figure~\ref{fig:Heston Sinf} for one set of parameters.

\begin{figure}\centering
\includegraphics[width=8cm]{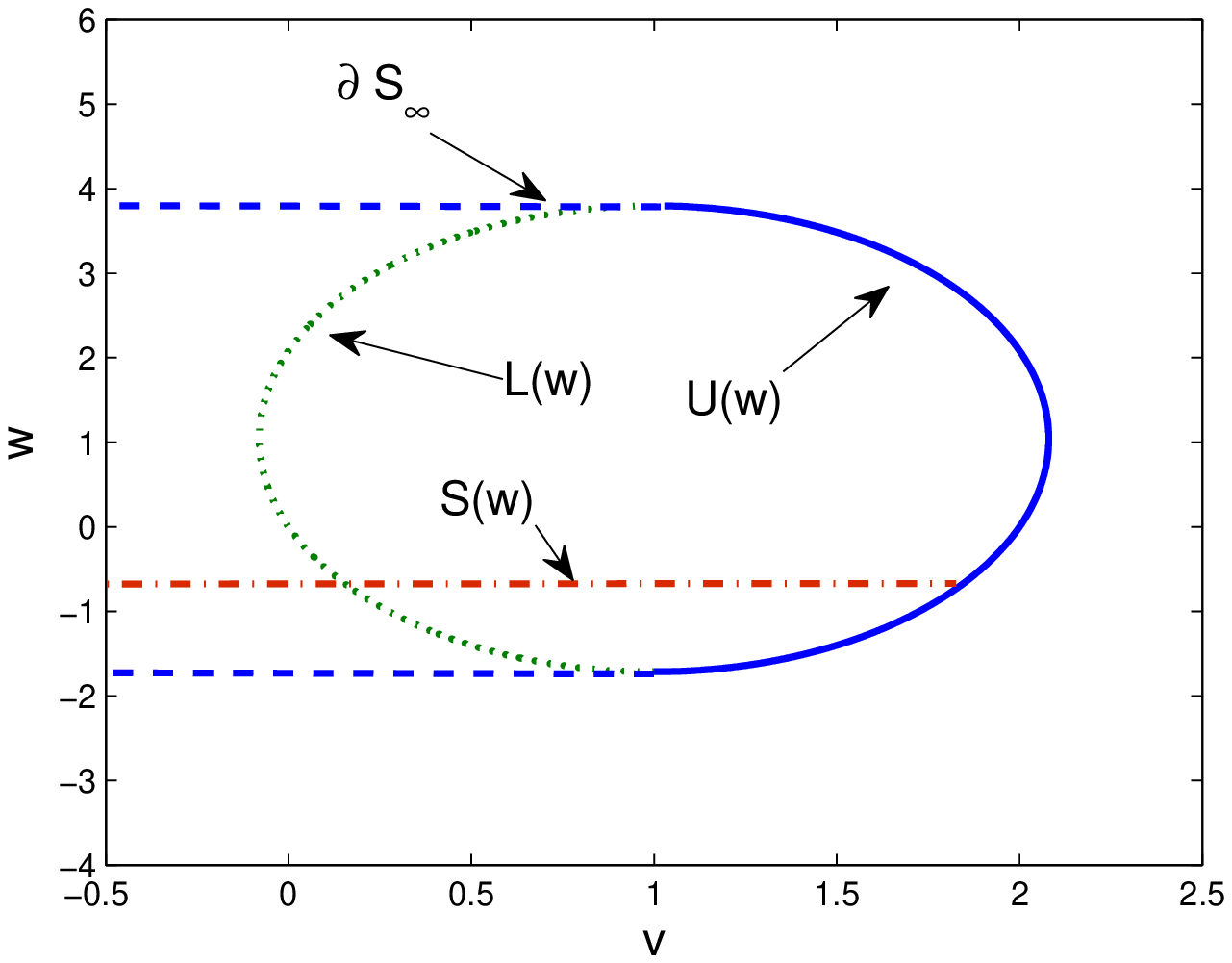}
\includegraphics[width=8cm]{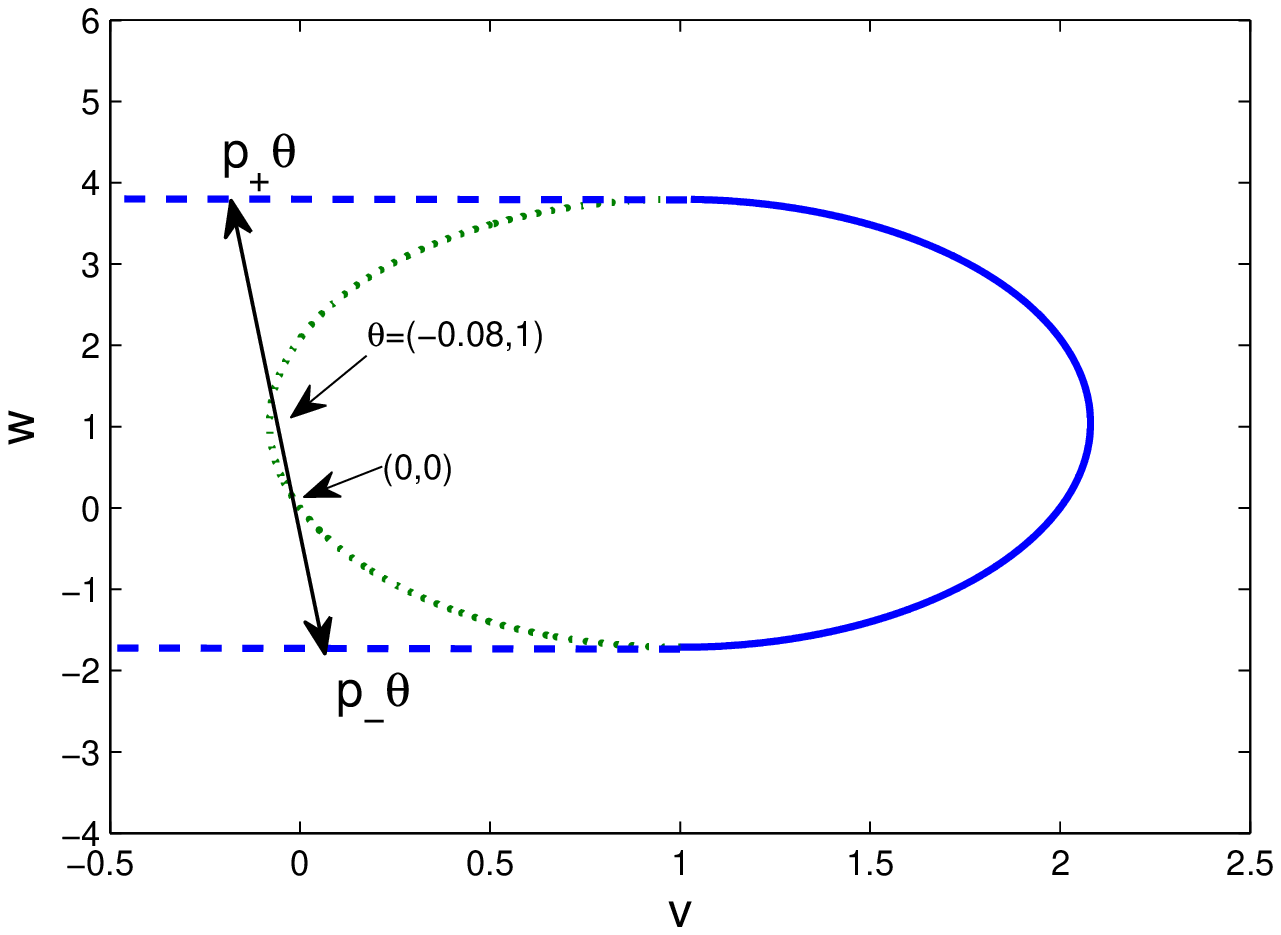}
\caption{Decomposition of $\mathcal{S}_\infty$ for the Heston model with $\kappa = 1$, $\sigma = 0.4$, and $\rho = -0.2$.}\label{fig:Heston Sinf}
\end{figure}

\begin{figure}\centering
\includegraphics[width=8cm]{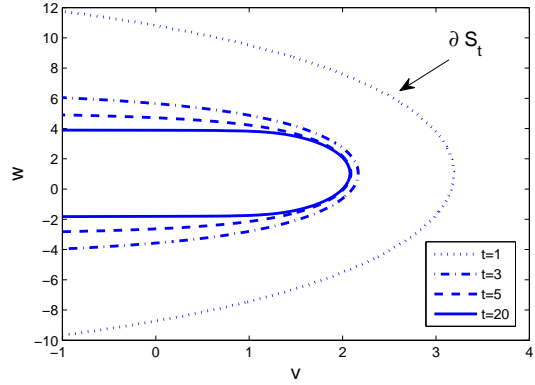}
\caption{$\partial \mathcal{S}_T$ for the Heston model with $\kappa = 1$, $\sigma = 0.4$, and $\rho = -0.2$.}\label{fig:HestonSt}
\end{figure}

Let us now comment on the parametric constraint $\kappa >\sigma \rho$ in \cite{FordeJ2}. This constraint is actually the necessary and sufficient condition for $\theta$ being a stable equilibrium point, in particular, $1\in \fD^\circ$. Indeed, recall that $\theta = (\sigma \rho, 1)^\top$ is an equilibrium point, hence $\theta \in \mathcal{S}_\infty$ and in particular $1\in \fD$. When $\kappa= \sigma \rho$, $\kappa^2-2g(1)=0$, then $1\in \partial\fD$. When $\kappa<\sigma\rho$, $U(1)= \kappa + |\kappa-\sigma \rho| = \sigma \rho$ hence $\theta\in \partial\mathcal{S}_\infty$ and unstable.

Under the parametric constraint in the last paragraph, extending $\theta$ in both directions until it reaches $\partial \mathcal{S}_\infty$, we obtain $p_{\pm} \theta$ with $p_+>1$ and $p_-<0$. This is illustrated in the right panel of Figure~\ref{fig:Heston Sinf}. Now it follows from Theorem \ref{thm: S-inf} that $\expec[S_T^p]$ is finite for any $T \geq 0$ and $p \in [p_-, p_+]$. Actually, Proposition \ref{thm: long term} implies that
\begin{equation*}
\Lambda(p):= \lim_{T\rightarrow \infty}\frac{1}{T}\log\mathbb{E}\exp\left(p\theta^\top X_T\right) = -\frac{\kappa\varphi\rho}{\sigma}p + \frac{\kappa\varphi}{\sigma^2}L(p), \quad p\in (p_-, p_+).
\end{equation*}
This result coincides with  Theorem 2.1 in \cite{FordeJ2} where the authors continue to prove the essential smoothness of $\Lambda(\cdot)$ and derive formulae for the large-time-to-maturity implied volatilities. On the other hand, in \cite{Keller-Ressel}, the author provides the same formula under the same constraint and argues that the price process gets close to a NIG L\'evy model, in terms of marginal distributions.

Figure~\ref{fig:HestonSt} shows $\partial \mathcal{S}_T$ for several different $T$ values. From this figure, we can identify critical exponents $p^*$ and $q^*$ for fixed $T$ as the first positive numbers $p$ and $q$ such that $(p+1)\theta$ and $-q\theta$ belong to $\partial \mathcal{S}_T$. We can also clearly see the convergence of $\mathcal{S}_T$ to $\mathcal{S}_\infty$ as $T \to \infty$. From the viewpoint of Theorem~\ref{thm:St-bdry}, $u \in \partial\mathcal{S}_T$ translates into the condition that $(v, 1)$ is the initial condition $x(0)$ such that $\lim_{s\uparrow \infty} x(s) = (2/T,0)$ where $x(\cdot)$ solves \eqref{eq:quadODE}. Therefore, as in Figure~\ref{fig:St-ex}, one finds $v$ such that $(v,1)$ is on the boundary of the stable set of \eqref{eq:quadODE}, yielding $(v,1) \in \partial\mathcal{S}_T$. However, we also note that the Heston model admits a closed form formula for $\partial\mathcal{S}_T$ by which implied volatilities at extreme strikes can be calculated. This theorem, nevertheless, provides one method of accomplishing the same task even when such a closed form formula is not available.

\begin{figure}\centering
  \includegraphics[width=8cm]{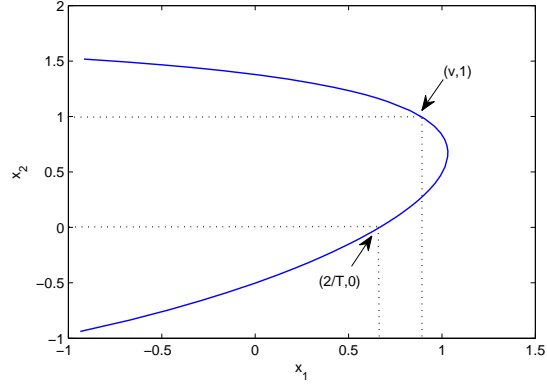}
  \caption{The stable boundary of \eqref{eq:quadODE} with $w=5$ and $T=3$.}\label{fig:St-ex}
\end{figure}

\subsection{Double stochastic volatility model}
The following 3-dimensional stochastic volatility model was proposed in \cite{Gatheral}:
\begin{eqnarray*}
dV_{t} &=& \kappa_1(V_{t}^{'}-V_{t})dt + V_{t}^{\alpha}dZ^{1}_t, \\
dV_{t}^{'} &=& \kappa_2(\varphi -V_{t}^{'})dt + V_{t}^{'\beta}dZ^{2}_t, \\
dS_t &=& S_t\sqrt{V_{t}}dZ^3_{t},
\end{eqnarray*}
where $Z_t$ is a correlated 3-dimensional Brownian motion and $\kappa_1, \kappa_2, \varphi$ are strictly positive constants. In this model, $V$ models the high-frequency  variance and $V^{'}$ represents the low-frequency variance.
For this model to be an affine diffusion, it is necessarily that $\alpha = \beta = 1/2$ and $Z^2$ is independent of $Z^1$ and $Z^3$. Denote the correlation between $Z^1$ and $Z^3$ as $\rho \in (-1,1)$. We assume that $\kappa_1 > \kappa_2$ so that the mean reverting speed of high-frequency variance is larger than its low-frequency analogue.

It is an easy matter to check that $Y_t = (V_t, V_t', \log S_t)$ satisfies
\[
dY_t = \left(\left( \begin{array}{c} 0 \\ \kappa_2\varphi \\ 0 \end{array}\right) +
      \left(  \begin{array}{ccc}
            -\kappa_1 & \kappa_1 & 0 \\
            0 & -\kappa_2 & 0 \\
            -1/2 & 0 & 0 \end{array}\right) Y_t \right)dt +
                \left( \begin{array}{ccc}
                        \sqrt{Y^1_t} & 0 & 0 \\
                        0 & \sqrt{Y^2_t} & 0 \\
                        \rho \sqrt{Y^1_t} & 0 & \sqrt{(1 - \rho^2)Y^1_t} \end{array}\right) dW_t,
\]
 where $W$ is a 3-dimensional standard Brownian motion. Now, applying an appropriate linear transform $X = \Lambda Y$, we obtain from straightforward calculations that $X$ satisfies $dX_t = (\hat{b}+\hat{B}X_{t})dt+\hat{\sigma}(X_{t})dW_{t}$ with
\[
 \hat b = \left( \begin{array}{c}
	0 \\
            \kappa_2\varphi\\
            0 \end{array}\right), \;\;
 \hat B = \left( \begin{array}{ccc}
	 -\kappa_1  & \kappa_1 & 0 \\
  	 0 & -\kappa_2 & 0 \\
	\rho\kappa_1 - 1/2 & -\rho\kappa_1 & 0 \end{array}\right), \;\;
 \hat a(x) = \left(
\begin{array}{ccc}
 x_1 & 0 & 0 \\
 0 & x_2 & 0 \\
 0 & 0 & (1 - \rho^2)x_1
\end{array}
\right),
 \]
and the initial condition is given by $X_0=(V_0, V_0', -\rho V_0 + \log S_0)$. In addition, $\log S_T = Y_T^3 = (0\; 0\; 1)\Lambda^{-1}X_T$, which implies that $\theta = (\rho, 0, 1)$.

The associated Riccati system \eqref{eq:riccati} is
\begin{equation} \label{eq:riccatimulti}
\begin{split}
&\dot y_1 = \frac12 y_1^2 - \kappa_1y_1 + g_1(z), \\
&\dot y_2 = \frac12 y_2^2 -\kappa_2 y_2 + \kappa_1 y_1 + g_2(z), 
\end{split}
\end{equation}
and $\dot z = 0$ with the initial condition $(y_1(0), y_2(0), z(0)) = (u_1, u_2, w)$. Here, $g_1(z) =  0.5(1-\rho^2) z^2 + (\rho\kappa_1 - 0.5)z$ and $g_2(z) = -\rho \kappa_1 z$. Clearly, $\operatorname{Ker} A^\D = \Real$. Consider the following sets defined in Section \ref{subsec: stable eq}:
\[
\begin{split}
&\fE(w) = \left\{(u_1, u_2) \such (u_1 - \kappa_1)^2 \leq \kappa_1^2 -2 g_1(w), \; (u_2- \kappa_2)^2 \leq \kappa_2^2 - 2 \kappa_1 u_1 - 2 g_2(w) \right\},\\
&\fD = \set{w \in \operatorname{Ker} A^\D \such \fE(w) \neq \emptyset} = \set{w\in \operatorname{Ker} A^\D \such \kappa_1^2 - 2g_1(w) \geq 0, \kappa_2^2 - 2 g_2(w) \geq 2 \kappa_1 \eta_1(w)}
\end{split}
\]
with $\eta_1(w) := \kappa_1 - \sqrt{\kappa_1^2 - 2g_1(w)}$. Then,  Proposition~\ref{thm: long term} gives us
\[
 \Lambda(p):=\lim_{T\rightarrow \infty} \frac{1}{T} \log\expec[S_T^p] = \kappa_2 \varphi \eta_2(p) \quad \text{ for } p\in \fD^\circ, \quad \eta_2(p) := \kappa_2 - \sqrt{\kappa_2^2 - 2 g_2(p) - 2\kappa_1 \eta_1(p)}.
\]

Figure~4.3 shows the sets $\partial \mathcal{S}_\infty$  and  $\partial\mathcal{S}_T$ for some specific set of parameters.
Two graphs are obtained by numerically solving \eqref{eq:riccatimulti} and \eqref{eq:quadODE}. Theorems \ref{thm: S-inf} and \ref{thm:St-bdry} help to identify $\partial \mathcal{S}_\infty$  and  $\partial\mathcal{S}_T$ numerically, because it suffices to find stable sets for unstable equilibria of their associated systems.  To be more specific, for $\partial \mathcal{S}_\infty$, we first locate all unstable equilibrium points of \eqref{eq:riccatimulti} for fixed $w$. Then for each unstable $\nu$, we solve \eqref{eq:riccatimulti} backward in time to retrieve its stable set. Finally, $\partial \mathcal{S}_\infty$ is obtained by patching all $\partial \mathcal{S}_\infty(w)$ together. The right panel of Figure~\ref{fig:doubleheston} shows $\partial \mathcal{S}_T$ for $T=0.5$ and $1$. It is produced similarly by working with \eqref{eq:quadODE}.
Even though we only show part of $\partial \mathcal{S}_T$ in the right panel, it is understood that the critical exponents $p^*+1$ and $-q^*$ are the intersections of the $w$-axis and $\partial\mathcal{S}_T$.

\begin{figure}\centering
\includegraphics[width=8cm]{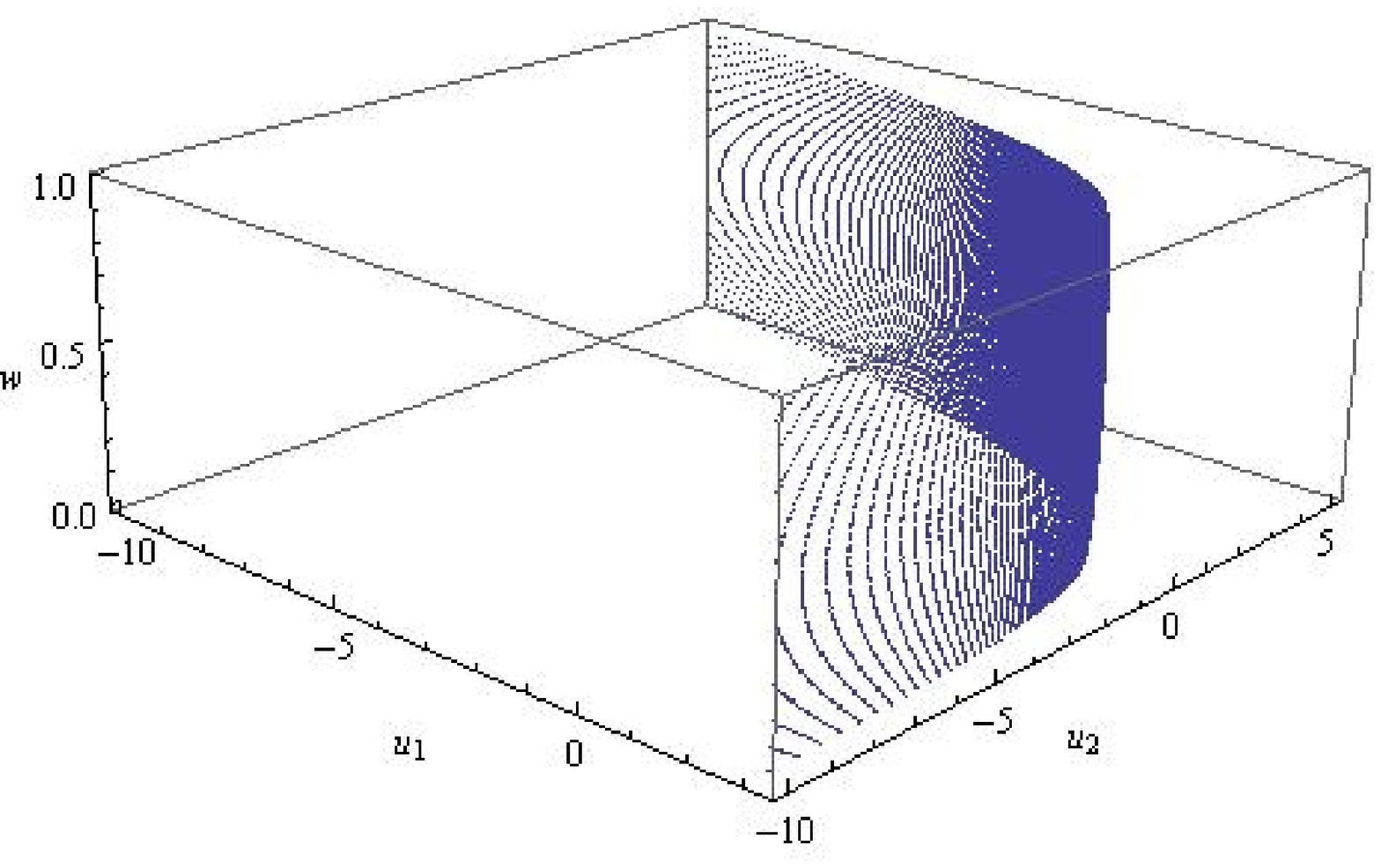}
\includegraphics[width=8cm]{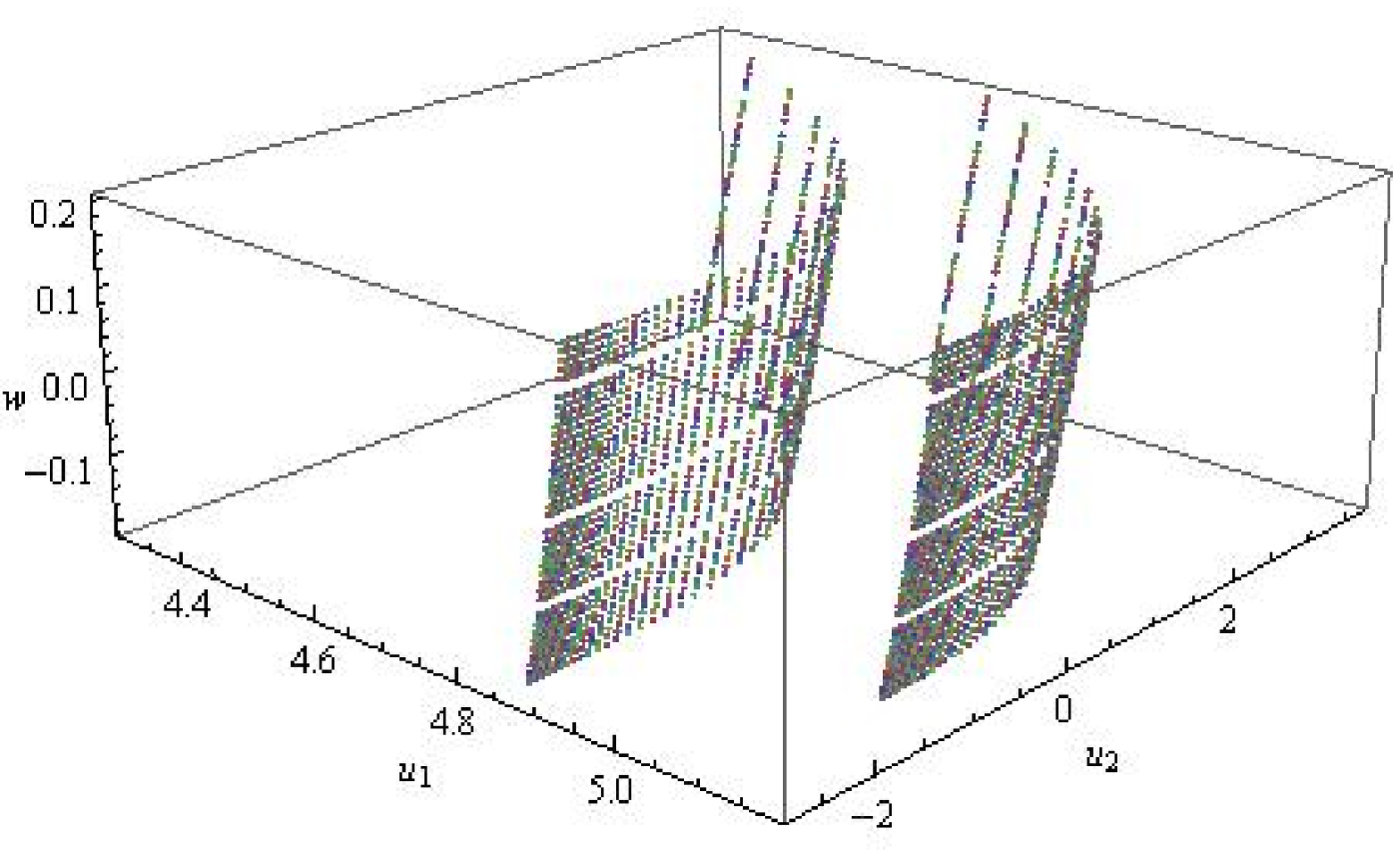}
\caption{Region $\mathcal{S}_{\infty}$ and $\mathcal{S}_T$ for $T=0.5$, $T=1$ with $\kappa_1=2, \kappa_2=1$.}\label{fig:doubleheston}
\end{figure}

To derive the large-time-to-maturity implied volatilities in this model, we need the following parameter restriction:
\[\kappa_1 >\rho.\]
As we have seen in the previous subsection, this restriction ensures that $\theta$ is a stable equilibrium point of \eqref{eq:riccatimulti}. As a result, $[0,1]\subseteq \fD^\circ$. On the other hand, one can check that $\Lambda$ is essentially smooth in $\fD$. Indeed,
\[
 \Lambda'(p) = \kappa_1 \kappa_2 \varphi \times \frac{-\rho + g_1'(p)/\sqrt{\kappa_1^2 - 2g_1(p)}}{\sqrt{\kappa_2^2 - 2g_2(p)-2\kappa_1 \eta_1(p)}},
\]
and both terms under the square roots converge to zero as $w\rightarrow \partial \fD$. Consequently, the G\"{a}rtner-Ellis Theorem applies, and thus $\{(\log S_t - \log S_0)/t\}$ satisfies the Large Deviation Principle under $\prob$ with the rate function (Legendre transform) $\Lambda^*(x) = \sup_{p\in\fD}\{xp - \Lambda(p)\}$. It then follows from Proposition 4.1.3 in \cite{Jacquier} that the large-time-to-maturity implied volatility for the European option with maturity $T$ and strike price $K(T) = S_0 \exp(xT)$ is
\begin{equation}\label{eq: sigma_inf}
 \sigma^2(x, \infty):= \lim_{T\rightarrow \infty} \sigma^2(x, T) = 2 \pare{2 \Lambda^*(x) -x + 2\pare{\indic_{\{x\in (x^*, \tilde{x}^*)\}}-\indic_{\{x\notin (x^*, \tilde{x}^*)\}}}\sqrt{\Lambda^*(x)^2- x\Lambda^*(x)}},
\end{equation}
where $x^*:= \Lambda'(0)= -\varphi/2$ and $\tilde{x}^*:= \Lambda'(1)=0.5 \kappa_1 \varphi/(\kappa_1 -\rho)$. In particular, the large-time-to-maturity implied volatility for at-the-money European option is
\[
 \sigma^2(0, \infty) = 8 \Lambda^*(0) = -8 \Lambda(p_0), \quad  p_0 = \frac{1-2\rho \kappa_1 + |\rho|\sqrt{1+4\kappa_1^2 - 4\rho \kappa_1}}{2(1-\rho^2)}.
\]
Here, $p_0$ is chosen so that $\Lambda'(p_0)=0$. On the other hand, we can also obtain the leading order expansion of $\sigma(x,\infty)$ when $x$ is close to $0$. This leading order expansion provides us information on the implied volatility asymptotics for fixed strikes. This is because one can choose $x= T^{-1}\log (K/S_0)$ for the fixed strike $K$. When $T$ is large, $x$ is close to zero. To obtain this expansion, one first observes that
$\Lambda^*(x) = \Lambda^*(0) + x (\Lambda^*)'(0) + o(x) = \Lambda^*(0) + p_0 x + o(x)$ where the second identity follows from $(\Lambda^*)'(0)= (\Lambda')^{-1}(0)=p_0$. Plugging the previous expansion into \eqref{eq: sigma_inf}, we get
\[
 \sigma^2(x,\infty) = 8 \Lambda^*(0) + 2(2p_0-1) \sqrt{x} + o(\sqrt{x}).
\]

\subsection{Cascading affine diffusions}
To take the multi-frequency aspect of interest rates into account, \cite{Wu10} consider a model in which the interest rate has the dynamics that depends on several latent variables with high to low frequencies. The authors also suggest a multi-frequency stochastic volatility model for equity option pricing. In this subsection, we consider a specific form of such cascading volatility models which are also affine diffusions.

In this model, $Y^\D$ has no restriction other than the admissibility constraints on its parameters, while $Y^\V$ follows
\begin{equation*}\label{eq:cascading}
\begin{split}
& dY^i_t = \kappa \delta^{i-1}\left( Y^{i+1}_t - Y^i_t\right) dt + \sigma \sqrt{Y^i_t}dW^i_t, \quad i=1, \ldots, m-1\\
& dY^m_t = \kappa\delta^{m-1}\left( \varphi - Y^m_t\right) dt + \sigma \sqrt{Y^m_t}dW^m_t
\end{split}
\end{equation*}
with $0 <\delta < 1$ and positive constants $\kappa$, $\sigma$, and $\varphi$. This model proposes that the volatility process $Y^1$ depends on many latent variables that have slower mean reversion speeds. Then, the process $X = \Lambda Y$ with
$
\Lambda = \left( \begin{array}{cc}
                \sigma^{-2} I_m & 0 \\
                0 & I_n \end{array}\right)
$
makes $X$
a canonical version of
$Y$. The associated Riccati system \eqref{eq:riccati} reads
\[
\dot y_i = \frac{1}{2} y_i^2  + \kappa\left(\delta^{i-2}y_{i-1}\indic_{i>1} - \delta^{i-1}y_i\right) + \frac{1}{2}{z}^\top \pi_i z + \left(A^\C z\right)_i, \quad i=1,\ldots, m,
\]
and $\dot z = A^\D z$ with the initial condition $(y(0), z(0)) = u \in \mathbb{R}^m \times \operatorname{Ker} A^\D$.

Let us consider a simple case where $A^\D$ is invertible. Then, we easily see that
$\operatorname{Ker} A^\D = \{0_n\}$ with $n$-dimensional zero vector $0_n$ and that the equilibrium points are the origin, i.e., $\eta(0_n) = 0\in \Real^{m}$, and
$\nu = (0, \ldots, 0, 2\kappa\delta^{m-1})$. Moreover, the Jacobian matrices of equilibrium points are given by
\[
\left\{ \begin{array}{l}
    J(0)_{i, i} = - \delta^{i-1}\\
    J(0)_{i+1, i} = \delta^{i-1}
    \end{array}\right.
, \quad
\left\{ \begin{array}{l}
    J(\eta)_{i, i} = - \delta^{i-1}\indic_{i < m} + \delta^{m-1}\indic_{i=m} \\
    J(\eta)_{i+1, i} = \delta^{i-1}
    \end{array}\right..
\]
Therefore, $\nu$ is a hyperbolic equilibrium point and its stable submanifold $W^s_\nu$ has dimension $m-1$. Theorem~\ref{thm: S-inf}, then, implies that $\mathcal{S}_\infty^\circ = \mathcal{S}(0_n)\times\{0_n\}$ and
$\partial \mathcal{S}_\infty = W^s_\nu(0_n) \times \{0_n\}$. As a consequence, we obtain from Corollary~\ref{cor: exp moment X} that
$$
\lim_{T \rightarrow \infty}\frac{1}{T}\log \mathbb{E}\left[\exp\left(u^\top X_T\right)\right]
= 2\varphi\left(\frac{\kappa \delta^{m-1}}{\sigma}\right)^2
$$
for $u \in \partial \mathcal{S}_\infty$.
As for the blow-up region, we have
$\partial \mathcal{S}_T = \left(\bigcup_{\nu \neq 0} W^s_\nu(0, T)\right) \times \{0\}$ where $\nu$ is any $m$-dimensional vector whose entries are either 0 or $2/T$.

Lastly, let us consider an equity model based on cascading affine diffusions. We just add one last process as follows:
$$
dY^{m+1}_t = -\frac{1}{2}Y^1_t dt + \sqrt{Y^1_t}dW^{m+1}_t
$$
with a standard Brownian motion independent of $(W^1, \ldots, W^m)$, and set $S_t = \exp(Y^{m+1}_t)$ so that $\mathbb{E}S_t^p = \mathbb{E}\exp(p\theta\cdot X_t)$ where $\theta = (0, \ldots, 0, 1)$. This can be understood as a general version of the Heston model but with no correlation between Brownian parts for the sake of simplicity of exposition. The associated Riccati system can be calculated accordingly.
From Proposition~\ref{thm: long term}, we have $\Lambda(p):=\lim_T T^{-1}\log\mathbb{E}S_T^p = \kappa\varphi\delta^{m-1}\eta_m(p)/\sigma^2$ if $p\theta \in \mathcal{S}_\infty^\circ$ and $\nu_m$ if $p\theta\in \partial\mathcal{S}_\infty$ where $\nu$ is some unstable equilibrium point. The stable equilibrium point of the above quadratic system can be found iteratively as follows:
\begin{eqnarray*}
  \eta_1(p) &=& \kappa - \sqrt{\kappa^2 - \sigma^2 p(p-1)}, \\
  \eta_i(p) &=& \kappa\delta^{i-1} - \sqrt{\kappa^2\delta^{2(i-1)} - 2\kappa\delta^{i-2}\eta_{i-1}(p)}, \quad i=2, \ldots,m
\end{eqnarray*}
where $p$ belongs to some interval, say $[a, b]$, so that all the square root terms are well-defined. Such conditions read $p \in [p_-,p_+]$ with $p_\pm = \left(\sigma \pm \sqrt{\sigma^2 + 4\kappa^2}\right)/(2\sigma)$, and $\eta_{i-1}(p) \leq \kappa\delta^i/2$ for $i=2,\ldots, m$. It is useful to check that this interval gets strictly smaller as $i$ increases, which we leave it as a simple exercise. Moreover, we get
$$
\Lambda'(p) = \frac{\kappa\varphi\delta^{m-1}}{\sigma^2}\times \frac{\kappa\delta^{m-2}\eta_{m-1}'(p)}{\sqrt{\kappa^2\delta^{2(m-1)}-2\kappa\delta^{m-2}\eta_{m-1}(p)}}.
$$
At the boundary points of $[a, b]$, the square root term converges to zero. Hence, $|\Lambda'(p)|\rightarrow \infty$ and thus $\Lambda(p)$ is essentially smooth. By following the same arguments as in the previous subsection, we can obtain the implied volatility asymptotic formula at large-time-to-maturities.

\section{Analysis of the long-term behavior}\label{sec:moment}
The long-term distributional properties of affine processes are determined by the long-term behaviors of solutions for the associated Riccati system. Therefore, we shall first focus on equilibrium analysis of \eqref{eq:riccati} and prove Lemma~\ref{lem: eq point summary} in Section~\ref{subsec: stable eq} and then characterize the long-term behavior of its solutions in Section~\ref{subsec: stable reg}. Finally, Theorem \ref{thm: S-inf} and Corollary \ref{cor: exp moment X} are proven at the end of this section.

\subsection{Stable equilibrium points}\label{subsec: stable eq}
Let us start with some definitions inspired by \cite{Keller-Ressel}. Define the following two sets:
\[
\fE := \set{u \in \Real^d \such f_i(u)\leq 0, \forall\, 1\leq i \leq m}\quad \text{ and } \quad \fE^\circ := \set{u\in \Real^d \such f_i(u)< 0, \forall\, 1\leq i \leq m}.
\]
They are sets of points on which all components of $f$ are simultaneously (strictly) negative.
It follows from the continuity of $f$ that $\fE$ ($\fE^\circ$) is closed (open), respectively. Moreover, they are convex thanks to the convexity of $f$. It is also clear that $\fE$ is nonempty, since $0\in \fE$. Given $w \in \operatorname{Ker} A^{\D}$, we define sections of $\fE$ and $\fE^\circ$ as
$\fE(w) := \set{v\in \Real^m \such (v, w) \in \fE}$ and $\fE^\circ(w) := \set{v\in \Real^m \such (v, w) \in \fE^\circ}$.
To identify all stable equilibrium points for \eqref{eq:riccati-V}, we define
\begin{equation}\label{def: D}
  \fD :=\set{w \in \operatorname{Ker} A^{\D} \such \fE(w) \neq \emptyset} \quad \text{ and } \quad \fD^\circ := \set{w \in \operatorname{Ker} A^{\D} \such \fE^\circ(w) \neq \emptyset}.
\end{equation}
It will be shown in Lemma \ref{lem: dense in fD} blow that $\fD^\circ$ is indeed the interior of $\fD$. The first result below identifies the candidate stable equilibrium point for \eqref{eq:riccati-V}.

\begin{lem}\label{lem: eq point}
 Given $w \in \fD$, $\eta(w)$, defined in \eqref{eq: def eta}, is inside $\fE$. Moreover, $f(\eta(w), w) = 0$ and $\eta(w) \leq v$ for any $v \in \fE(w)$.
\end{lem}
\begin{proof}
 We utilize the lower triangular shape of $A^{\V}$ and prove the statement by induction on $i$. This type of argument will be used repeatedly in our analysis.

 For $i=1$, if $i\notin \I$, $f_1$ is linear with slope $A_{11}<0$, then $\eta_1(w)$ is chosen as the solution to $f_1(\cdot, w)=0$. Clearly $\eta_1(w)\leq v_1$ for any $v\in \fE(w)$. If $i \in \I$, the quadratic equation $f_1(\cdot, w)=0$ has solution(s) because the graph of $f_1(\cdot, w)$ has a nonempty intersection with $\mathbb{R}\times \mathbb{R}_-$. Then $\eta_1(w)$ is chosen as the smaller of the two solutions (possibly the same) to the previous quadratic equation. It is also clear that $\eta_1(w)\leq v_1$ for any $v\in \fE(w)$.

 Suppose now that the statement holds for $k=1, \cdots, i-1$. If $i\notin \I$, it then follows from $A_{ii}<0$ and $A_{ik}\geq 0$ for $k\neq i$ that, for any $v\in \fE(w)$,
 $
 \eta_i(w) = -A_{ii}^{-1}\left( \sum_{k=1}^{i-1}A_{ik}\eta_k(w) + g_i(w)\right) \leq -A_{ii}^{-1}\left( \sum_{k=1}^{i-1}A_{ik}v_k + g_i(w)\right) \leq v_i,
 $
 where the second inequality holds since $f_i(v,w)\leq 0$. Now if $i\in \I$, notice that
 $
 v_i^2/2 + A_{ii}v_i + \sum_{k=1}^{i-1}A_{ik}\eta_k(w) + g_i(w) \leq v_i^2/2 + A_{ii}v_i + \sum_{k=1}^{i-1}A_{ik}v_k + g_i(w)\leq 0
 $, for any $v\in \fE(w)$. Then $\eta_i(w)$, which is the smaller root of the quadratic function in $v_i$ on the left side of above inequalities, must be less or equal to $v_i$. Hence the induction step is proved. Finally, $\eta(w)\in \fE(w)$ is clear from the construction of $\eta(w)$.
\end{proof}

Note that the Jacobian of $f$ at $\eta(w)$ is $A^{\V}+ diag_{\I} (\eta(w))$. A straightforward induction on the index $i$ shows that $w\in \fD^\circ$ if only only if $A^{\V} + diag_{\I}(\eta(w))$ has strictly negative diagonals. In order to show that $\eta(w)$ is indeed the unique asymptotically stable equilibrium point, let us present some topological properties of $\fD$ and $\fD^\circ$ in the next two lemmas.
\begin{lem}\label{lem: top prop}
The set $\fD$ ($\fD^\circ$) is a nonempty closed (open) convex subset of $\operatorname{Ker} A^{\D}$. Also, $\eta$ is a convex function on $\fD$ and $\eta\in C(\fD) \cap C^2(\fD^\circ)$.
\end{lem}

\begin{proof}
 Let us first show $0\in \fD^\circ$. Indeed, for any $v\in \fE(0)$, $f(v,0) = \frac12 v^{(2)}_{\I} + A^{\V} v \leq 0$, and thus $v \geq - (1/2)(A^{\V})^{-1} v^{(2)}_{\I} \geq 0$, where the last two inequalities follow from the nonsingular M-matrix property of $-A^{\V}$ (see Definition~\ref{def: M-matrix}). Hence, Lemma~\ref{lem: eq point} implies that $\eta(0)$ must be $0$. Notice that the Jacobian of $f(\cdot, 0)$ at $0$ is $A^{\V}$, which has strictly negative diagonals. Therefore, $0 \in \fD^\circ \subset \fD$.

 For the convexity of $\fD$, it suffices to show that $\hat w:= \lambda w+ (1-\lambda) \tilde{w} \in \fD$ for any $w,\tilde{w}\in \fD$ and $\lambda\in [0,1]$. Actually, since $f(\eta(w), w) = f(\eta(\tilde{w}), \tilde{w})=0$, it then follows from the convexity of $f$ that $f(\lambda \eta(w)+(1-\lambda)\eta(\tilde{w}), \hat w)\leq 0$. This implies $\fE(\hat w) \neq \emptyset$, hence $\hat w \in \fD$. Therefore we get $\eta(\hat w)$ from Lemma~\ref{lem: eq point} and $\eta(\hat w) \leq \lambda \eta(w)+(1-\lambda)\eta(\tilde{w})$ as $\eta(\hat w)$ is the componentwise minimum of $\fE(\hat w)$.

 The convexity of $\fD^\circ$ is confirmed once we show $\hat{w}\in \fD^\circ$ for any $w$, $\tilde{w}\in \fD^\circ$.
 When $w$ and $\tilde{w}$ are in $\fD^\circ$, both $A^\V+ diag_\I(\eta(w))$ and $A^\V + diag_\I(\eta(\tilde{w}))$ have strictly negative diagonals. We deduce from the convexity of $\eta$ in $\fD$ that $A^\V + diag_\I(\eta(\hat w)) \leq \lambda \left(A^\V + diag_\I(\eta(w)\right) + (1-\lambda)\left(A^\V + diag_\I(\eta(\tilde w))\right)$. Strictly negative diagonal entries of the matrices on the right hand side imply that those of $A^\V+ diag_\I(\eta(\hat w))$ are also negative, which in turn implies $\hat w \in \fD^\circ$.

 That $\eta$ is continuous is immediate from its construction in Lemma~\ref{lem: eq point}. Note that the square root term in $\eta(w)$ is nonzero for any $w\in \fD^\circ$, hence $\eta\in C^2(\fD^\circ)$ follows. The continuity of $\eta$ implies that $\fD$ is closed. To see this, let us take a sequence $\{w_n\} \subset \fD$ such that $\lim_n w_n = w$. Apparently, $w \in \operatorname{Ker} A^\D$. Since $\eta$ is continuous,
 $\lim_{n\rightarrow\infty} \eta(w_n) = \eta(w)$ which is defined in \eqref{eq: def eta},
 resulting in $\fE(w) \neq \emptyset$ and thus $w \in \fD$.
 The openness of $\fD^\circ$ is obvious from the continuity of $f$.
\end{proof}

\begin{lem}\label{lem: dense in fD}
 If $\fE(w) \neq \emptyset$, then $\fE^\circ(\lambda w) \neq \emptyset$ for any $\lambda\in [0,1)$. Therefore, $\fD^\circ$ is the interior of $\fD$.
\end{lem}
\begin{proof}

  We begin with setting $\fE_i(w):= \set{v\in \Real^m \such f_i(v, w) \leq 0}$ and define $\fE_i^\circ$ similarly. Certainly, $\fE(w) = \bigcap_{i=1}^m \fE_i(w)$ (similarly for  $\fE_i^\circ$).

 Let us fix $w \in \fD$. The proof is by an induction on $i$. When $i=1$, if $1\notin \I$, $f_1(\cdot,w)$ is a linear function. Since $A_{11}<0$, it is clear that $\fE_1^\circ(\lambda w) \neq \emptyset$ for any $\lambda \in [0,1)$. If $1\in \I$, $f_1(\cdot, w)$ is quadratic, we denote the determinant of $f_1(\cdot, \lambda w)$ as
 $
  \Delta_1(\lambda) = -\lambda^2 w\top \pi_1 w -2\lambda (A^\C w)_1 + A_{11}^2,
 $
 which is either a linear or quadratic function of $\lambda$. It follows from $\fE_1(w) \neq \emptyset$ that $\Delta_1(1)\geq 0$. On the other hand, observe that $\Delta_1(0) = A_{11}^2>0$ and $w^\top \pi_1 w\geq 0$ ($\pi_1$ is semi-positive definite), we then obtain $\Delta_1(\lambda) >0$ for any $\lambda\in [0,1)$. This implies that $\fE_1^\circ(\lambda w)\neq \emptyset$ for any $\lambda\in [0,1)$.

 Next, assuming that $\bigcap_{k=1}^{i-1}\fE_k^\circ(\lambda w) \neq \emptyset$ for any $\lambda\in [0,1)$, we want to show that $\bigcap_{k=1}^i\fE_k^\circ(\lambda w) \neq \emptyset$ for any $\lambda \in [0,1)$. For any $k=1, \cdots, i-1$ and $\lambda\in [0,1)$, $\eta_k(\lambda w)$ is well defined by the induction assumption and Lemma~\ref{lem: eq point}.  Consider $v=(\eta_1(\lambda w), \cdots, \eta_{i-1}(\lambda w), v_i, \cdots, v_m)$. Then,
 \[
  f_i(v, \lambda w) = \frac12 v_i^2 \, \indic_{i\in \I} + A_{ii} v_i + \sum_{k=1}^{i-1} A_{ik} \,\eta_k(\lambda w) + \frac12 \lambda^2 w^\top \pi_i w + \lambda (A^\C w)_i
 \]
 is either a linear or quadratic function in $v_i$. As seen in the $i=1$ case, the linear case is easy to handle. Hence we consider the quadratic case only. In this case, the determinant of $f(\cdot, \lambda w)$ is given by
 \[
  \Delta_i(\lambda) = -\lambda^2 w^\top \pi_i w - 2\lambda (A^\C w)_i - 2\sum_{k=1}^{i-1} A_{ik} \eta_k(\lambda w) + A^2_{ii}.
 \]
 Since $A_{ik} \geq 0$, $\pi_i$ is semi-positive definite, and $\eta$ is convex (see Lemma~\ref{lem: top prop}), $\Delta_i(\lambda)$ is concave in $\lambda$. At $\lambda = 1$, we know that $\eta(w)$ exists and thus $f_i(v, w) = 0$ has a solution. Therefore, $\Delta_i(1) \geq 0$. This observation together with the concavity of $\Delta_i(\cdot)$ and $\Delta_i(0) = A_{ii}^2 > 0$ shows that $\Delta_i(\lambda)>0$ for any $\lambda\in [0,1)$. Hence, $\bigcap_{k=1}^i\fE_k^\circ(\lambda w)\neq \emptyset$, closing the induction.

 For the second statement, notice that for any $w$ in the interior of $\fD$, there exists $\tau>1$ such that $\tau w\in \fD$. Then, the first statement implies $\lambda\tau w \in \fD^\circ$ for any $\lambda \in [0, 1)$. Hence, with $\lambda = 1/\tau$, we see $w \in \fD^\circ$. Together with the openness of $\fD^\circ$, we conclude that $\fD^\circ$ is the interior of $\fD$.
\end{proof}

Since $A^{\V}$ is assumed to be lower triangular, the Jacobian $Df(\nu)$ at a hyperbolic equilibrium point $\nu$ has nonzero real eigenvalues. Moreover, it is well known that $\nu$ is unstable if $Df(\nu)$ has a positive eigenvalue, and $\nu$ is asymptotically stable if all eigenvalues of $Df(\nu)$ are negative. Consequently,
a hyperbolic point $\nu$
is (asymptotically) stable if and only if every eigenvalue of $Df(\nu)$ is negative. Now we are ready to prove Lemma \ref{lem: eq point summary}.

\begin{proof}[Proof of Lemma \ref{lem: eq point summary}]
 For $w\in \fD$, all equilibrium points are constructed by solving $f(v, w) = 0$ sequentially from index $i = 1$. Clearly, there are at most $2^{|\I|}$ equilibrium points for \eqref{eq:riccati-V}.

 When $w \in \fD^\circ$, we already observed that $A^\V + diag_\I(\eta(w))$ has strictly negative diagonals. Thus, $\eta(w)$ is hyperbolic and asymptotically stable. For the uniqueness, we simply note that, for another equilibrium point $v$, if we let $i$ be the first index such that $v_i > \eta_i(w)$, then by construction, it must be that $i \in \I$ and $v_i = -A_{ii} + \sqrt{ A_{ii}^2 - 2\left( \sum_{k=1}^{i-1}A_{ik}\eta_k(w) + g_i(w)\right)}$. Since the $i$-th diagonal entry of $Df(v)$ is $A_{ii} + v_i > 0$, hence $v$ is unstable which contradicts with the choice of $v$.

 There exists at least one hyperbolic unstable equilibrium point. Take $v$ such that its first $l-1$ components are equal to those of $\eta(w)$ but $v_l = -A_{ll} + \sqrt{ A_{ll}^2 - 2\left( \sum_{k=1}^{l-1}A_{lk}\eta_k(w) + g_l(w)\right)}$ where $l = \max\set{i \such i \in \I}$.

 For $w\in \operatorname{Ker} A^\D \cap \fD^c$, if there exists an equilibrium point $v$ for \eqref{eq:riccati-V}, then $(v, w)$  would be an equilibrium point for \eqref{eq:riccati}. Therefore, $v\in \fE(w)$ which contradicts to $w \in \fD^c$.
\end{proof}

\begin{exa} For an intuitive understanding of the proof above, let us take a look at the following two dimensional system:
\[
    \dot y = f(y) = \frac{1}{2}\left( \begin{array}{c}
                            y_1^2 \\
                            y_2^2 \end{array} \right) + \left( \begin{array}{cc}
                                            -a & 0 \\
                                            1/2 & -1 \end{array} \right)y, \quad a \geq 0.
\]
In this example, $\I = \set{1, 2}$ and $m = 2$. We set $g(\cdot) \equiv 0$ for the simplicity of illustration. Therefore all sets blow are independent of $w$. Figure~\ref{fig:2dex} shows the graphs of points that satisfy each of $f_i(y)=0$ for three different values of $a$. The intersections of solid and dash lines are equilibrium points.  The filled regions without boundaries correspond to their $\fE^\circ(w)$ (for any $w$). We note that the number of equilibrium points changes as $a$ varies. It is easy to check that the origin is hyperbolic whose Jacobian has negative eigenvalues as long as $a > 0$, or equivalently $\fE^\circ(w) \neq \emptyset$. At $a = 0$, $\fE^\circ(w)$ becomes empty (also $\fD^\circ = \emptyset$) and $(0,0)$ is not hyperbolic.

Observe that the point $(1,1)$ in the second panel is the only non-hyperbolic equilibrium point among all equilibria in three figures (and actually for all $a > 0$). Therefore, we can imagine that this kind of non-hyperbolic equilibrium case ``rarely'' happens. Indeed, the hyperbolicity of equilibria is a generic property of dynamical systems (see, e.g., Section IV of \cite{Chiang88} and references therein).

\begin{figure}\centering
\includegraphics[width=5.3cm]{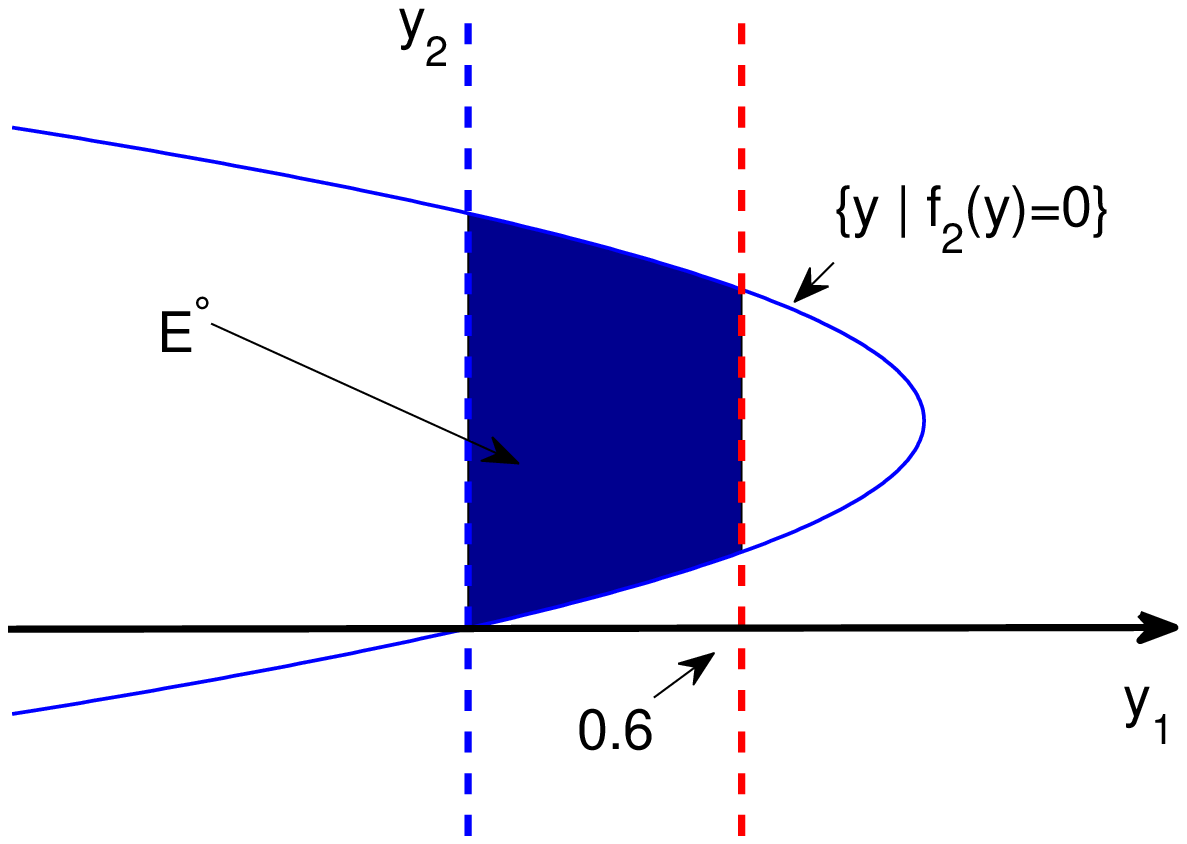}
\includegraphics[width=5.3cm]{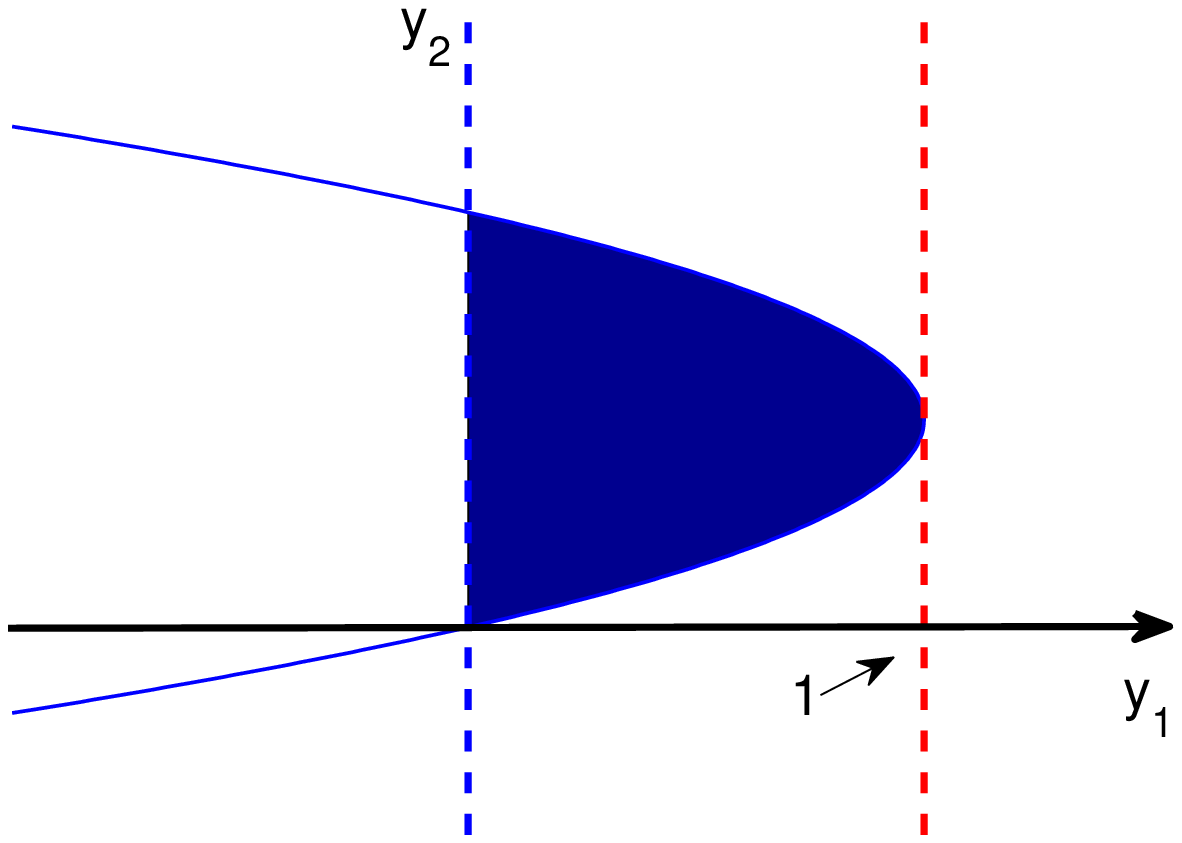}
\includegraphics[width=5.3cm]{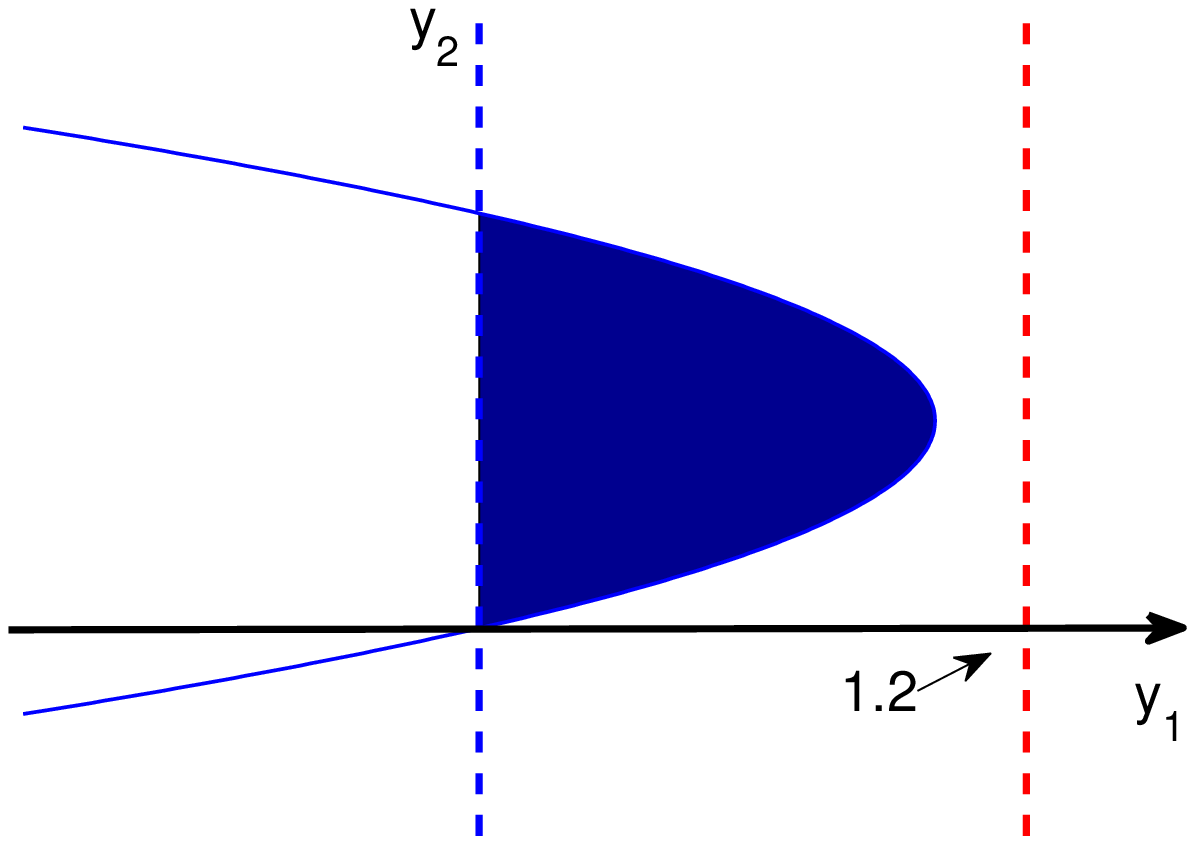}
\caption{Region $\fE^\circ(w)$ and equilibrium points for $a = 0.3, 0.5, 0.6$.}\label{fig:2dex}
\end{figure}
\qed
\end{exa}

\subsection{Stable regions}\label{subsec: stable reg}
After identifying the unique stable equilibrium point for \eqref{eq:riccati-V}, we will study its associated stable region defined in Definition~\ref{def: stable set}. When $w\in \fD^\circ$, the stable region $\mathcal{S}(w)$ is an open neighborhood of $\eta(w)$ diffeomorphic to $\mathbb{R}^m$ (see \cite{Chiang88}).
The next lemma gives a useful property on solution trajectories of \eqref{eq:riccati-V}.

\begin{lem}\label{lem:longtermproperty} If the trajectory of \eqref{eq:riccati-V} is bounded, then it converges to an equilibrium point. If the trajectory is unbounded, then it blows up in finite time.
\end{lem}
\begin{proof}  We use an induction starting from the index $i = 1$. But we present the induction step only, as $i=1$ case is a specification of the
general argument. Suppose that the trajectory of a solution $y(t)$ is bounded in $\mathbb{R}^m$ and that $y_k(t)$ for $k = 1, \ldots, i-1$ converges to the first $i-1$ components of some equilibrium point of \eqref{eq:riccati-V}. If $i \notin \I$, then $y_i(t)$ satisfies
$y_i(t) = e^{A_{ii}t}\left( y_i(0) + \int_0^t e^{-A_{ii}s}h(s)ds\right)$ where
$
h(s) = \sum_{k=1}^{i-1}A_{ik}y_k(s) + g_i(w)
$.
Since $A_{ii} < 0$ and $h(s)$ has a limit, say $h(\infty)$, it is easy to check that $\lim_t y_i(t) = -h(\infty)/A_{ii}$.

If $i \in \I$, then we have a quadratic ODE $\dot x(t) = x(t)^2/2 + \tilde{h}(t)$ where $x(t) = y_i(t) + A_{ii}$ and $\tilde{h}(t) = \sum_{k=1}^{i-1}A_{ik}y_k(t) + g_i(w) - A_{ii}^2/2$. The induction assumption says that $\tilde{h}(\infty)$ exists. Since the trajectory of $x(t)$ is bounded by assumption, Lemma~\ref{lem: bounded traj} implies that $x(\infty) := \lim_t x(t) \in \mathbb{R}^m$ and $x(\infty)^2/2 + \tilde{h}(\infty) = 0$. Then the limit of $y_i(t)$, together with $(y_1(t), \ldots, y_{i-1}(t))$, converges to the first $i$ components of some equilibrium point for \eqref{eq:riccati-V}.

Let us move onto the second statement. Assume that $i$ is the first index such that the trajectory of $y_i(t)$ is unbounded. It is necessarily that $i \in \I$. Let $x$ be defined as above: $\dot x = x^2/2 + \tilde{h}(t)$. By assumption, $\tilde{h}(t)$ is bounded (actually, converges to $\tilde{h}(\infty)$ as we have seen in the last paragraph) and $x(t)$ has unbounded trajectory. Therefore, we can find some time $t_0$ and nonnegative constant $C$ such that $x(t_0) > \sqrt{2C}$ and $\tilde{h}(t) \geq -C$. Let us consider a function $\tilde{x}(t)$ which solves $\dot {\tilde{x}} = \tilde{x}^2/2 - C$ for $t\geq t_0$ with $\tilde{x}(t_0) = x(t_0)$. Then, by the comparison theorem for scalar ODEs, we have $x(t) \geq \tilde{x}(t)$ for all $t \geq t_0$. However, $\tilde{x}(t_0)$ is greater than $\sqrt{2C}$ which is the unstable equilibrium point of the system for $\tilde{x}(t)$. Elementary computations show that $\tilde{x}(t)$ blows up in finite time, hence so does $x(t)$.
\end{proof}

\begin{rem}\label{rem: bdd below} Following similar arguments as in the proof of Lemma~\ref{lem:longtermproperty}, we can show that $y(t)$ is bounded below for any initial conditions.
\end{rem}

\begin{cor}\label{cor:outsideD} Let $w \in \operatorname{Ker} A^\D \cap \fD^c$. Then, any solution $y(t)$ to \eqref{eq:riccati-V} blows up in finite time.
\end{cor}
\begin{proof} If not, then the trajectory is bounded by the second statement of Lemma~\ref{lem:longtermproperty}. Then the first statement of the previous lemma implies that it should converge to an equilibrium point, say $v$. This contradicts with the choice of $w$.
\end{proof}

We now take on the task of characterizing $\mathcal{S}_{\infty}(w)$ for $w\in \fD$. Lemma~\ref{lem:longtermproperty} implies that
$
\mathcal{S}_\infty(w) = \bigcup_{\nu} W^s_\nu(w),
$
where $\nu$ ranges over all equilibria of \eqref{eq:riccati-V}. Dealing with $w \in \fD^\circ$ and $w \in \partial\fD$ cases separately, we  extract more information as presented in Lemmas~\ref{lem: S-inf int} and \ref{lem: S-inf bou}.

\begin{lem}\label{lem: S-inf int}
 Let $w\in \fD^\circ$. Then, $\mathcal{S}_\infty(w) = \overline{\mathcal{S}(w)}$ and $\partial \mathcal{S}_\infty(w) = \partial \mathcal{S}(w) = \bigcup_{\nu \neq \eta(w)}W^s_\nu(w)$ where $\nu$ is chosen from all equilibria of \eqref{eq:riccati-V}.
\end{lem}
\begin{proof} Part of the proof employs arguments in the proof of Theorem 4.1 of \cite{Kim}, but we include it here for completeness.
Instead of \eqref{eq:riccati-V}, it is more convenient to work with a slightly modified system on $r(t):= y(t)-\eta(w)$:
 \begin{equation}\label{eq: ODE w}
 \dot r = \frac12 r_\I^{(2)} + \hat A r, \quad \hat A := A^\V + diag_\I(\eta(w)).
 \end{equation}
 Since $w \in \fD^\circ$, $\hat A$ has strictly negative diagonals. Additionally, $A^\V$ is lower triangular with nonnegative off-diagonal entries, and therefore $-\hat A$ is a nonsingular M-matrix.
 It follows from Lemma \ref{lem: eq point summary} that $0$ is the unique stable equilibrium point of \eqref{eq: ODE w} and its associated stable region $\mathcal{S}'$ is given by
 \[
  \mathcal{S}' := \set{r(0) \such \lim_{t\to \infty} r(t) = 0 \text{ where } r \text{ solves } \eqref{eq: ODE w}} = \mathcal{S}(w) - \eta(w).
 \]
 Since $\mathcal{S}'$ is open, it contains a small neighborhood of $0$.
 The set $\mathcal{S}'_{\infty}$ is defined analogously. Then it suffices to prove the corresponding statements for $\mathcal{S}'$ and $\mathcal{S}'_{\infty}$.

 To show the first statement, we begin with the observation that $\mathcal{S}' \subset \mathcal{S}_\infty'$. Hence $\overline{\mathcal{S}'} \subseteq \mathcal{S}'_{\infty}$, since $\mathcal{S}_\infty' = \mathcal{S}_\infty(w) - \eta(w)$ is closed (see discussions after Lemma~\ref{lem:cont blowup time}). In what follows, we will prove $\mathcal{S}'_\infty \setminus \overline{\mathcal{S}'} =\emptyset$ by contradiction. Assume otherwise and pick $r(0) \in \mathcal{S}'_\infty \setminus \overline{\mathcal{S}'}$. It then follows from Lemma~\ref{lem:longtermproperty} that $\nu := \lim_{t\rightarrow \infty} r(t)$ is an equilibrium point. Since $\mathcal{S}'$ contains a neighborhood of $0$, we can find $\delta\in (0,1)$ such that $\delta r(0) \in \partial \mathcal{S}'$. Consider $z(t)$ which solves \eqref{eq: ODE w} with $z(0)= \delta r(0)$.
 Recall that $-\hat A$ is a nonsingular M-matrix, and then we get that $z(t)\leq \delta r(t)$ for any $t\geq 0$ from Lemma~\ref{lem: comparison}. On the other hand, since $z(0)\in \partial \mathcal{S}' \subset \mathcal{S}'_{\infty}$, Lemma~\ref{lem:longtermproperty} implies that $z(t)$ converges to an equilibrium, say $\tilde{\nu}$, of \eqref{eq: ODE w}. Hence, $\tilde\nu \leq \delta \nu$. Moreover, these two equilibrium points $\nu$ and $\tilde\nu$ are nonzero; otherwise, $r(t)$ or $z(t)$ should have started from $\mathcal{S}'$, which contradicts with their starting values outside $\mathcal{S}'$.

 Note that $\nu$ (also $\tilde \nu$) satisfies $\nu = -(1/2){\hat A}^{-1} \nu_\I^{(2)} \geq 0$ where the non-negativity follows from the property of a non-singular M-matrix $-\hat A$. Moreover, $0\leq\tilde\nu\leq \delta \nu$ implies that
 \[
  - \hat A \tilde\nu = \frac{1}{2} \,\tilde\nu_\I^{(2)} \leq \frac{\delta^2}{2} \, {\nu}^{(2)}_\I = -\delta^2\hat A\nu,
 \]
 from which we get $-\hat A (\tilde{\nu} - \delta^2 \nu) \leq 0$. By multiplying both sides by $-{\hat A}^{-1} \geq 0$, we obtain $\tilde\nu \leq \delta^2 \nu$.  Repeating the same argument, we arrive at $ \tilde \nu \leq \delta^{2k} \nu$ for any integer $k\geq 1$. However $\delta\in (0,1)$, hence $\tilde{\nu}=0$. This contradicts with $\tilde{\nu} \neq 0$ from the last paragraph. Therefore it is necessarily that $\mathcal{S}'_\infty \setminus \overline{\mathcal{S}'}$ is empty.

 Since $\partial \mathcal{S}'_\infty \cap \mathcal{S}' = \partial \mathcal{S}' \cap \mathcal{S}' = \emptyset$, it is immediately seen that $\partial \mathcal{S}' \subseteq \bigcup_{\nu \neq 0}{W^s_\nu}'$ where the ${W^s_\nu}'$ are the stable sets of each nonzero equilibrium point of \eqref{eq: ODE w}. To show the opposite inclusion, let us take an equilibrium point $\nu \neq 0$. Then, we claim that ${W^s_\nu}' \subseteq \partial \mathcal{S}'$. We prove this by contradiction. Let us assume that $r(0) \in {W^s_\nu}' \setminus \partial\mathcal{S}'$ and $\lim_{t\to \infty} r(t)= \nu$. Then it is necessarily that $r(0)\notin \mathcal{S}'$ because, otherwise, $\nu = \lim_t r(t)=0$. Hence, $r(0) \notin \overline{\mathcal{S}'}$ and thus we can find $\delta \in (0, 1)$ such that $\delta r(0) \in \partial \mathcal{S}'$. By using the same argument as that in the last paragraph, we arrive at a contradiction. Hence $\partial \mathcal{S}' = \bigcup_{\nu \neq 0}{W^s_\nu}'$.
 \end{proof}

When $w \in \partial \fD$, except in the case $m=1$, only partial or local description of $\mathcal{S}_\infty(w)$ is available as shown in the next result and arguments that follow. However, as a simple corollary to Theorem \ref{thm: S-inf}, $\mathcal{S}_\infty(w)$ for $w\in \partial \fD$ can be approximated by a limit of $\mathcal{S}_\infty(w_n)$ with $\set{w_n} \subset \fD^\circ$. To present the next result, we denote $(v_{i_1}, \ldots, v_{i_k})$ by $v_{\mathcal{M}}$, where
$\mathcal{M} = \set{i_1, \ldots, i_k} \subset \set{1, \ldots, m}$ and $v \in \mathbb{R}^m$.

\begin{lem}\label{lem: S-inf bou}
 For $w\in \partial\fD$, there exists a nonempty index set $\mathcal{M} \subset \set{1, \cdots, m}$ such that the set
 $\set{ {y(0)}_{\mathcal{M}} \such y(0) \in \mathcal{S}_\infty(w)}$ is equal to the stable set of a system
 \[
 \dot y_i = \frac12 y_i^2\indic_{i \in \I} + \sum_{k \in \mathcal{M}}A_{ik} y_k + g_i(w), \quad \text{ for } i \in \mathcal{M},
 \]
 which admits a unique equilibrium point.
\end{lem}

\begin{proof}
 When $w\in \partial\fD$, there exists some coordinate $i$ such that $\eta_i(w) = v_i$ for any $v\in \fE(w)$. If not, then for each $i = 1, \ldots, m$ there exists $v\in \fE(w)$ with $v_i> \eta_i(w)$. A simple induction argument shows that $A^\V + diag_\I(\eta(w))$ has strictly negative diagonals in this case. This contradicts with the choice of $w$. The set $\mathcal{M} = \set{i\in \set{1, \cdots, m} \such \eta_i(w) = v_i \text{ for any } v\in \fE(w)}$ is, then, nonempty.

 We claim that $A_{ij}=0$ for any $i>j$ such that $i\in \mathcal{M}$ and $j\in \mathcal{M}^c$. Since $A_{ij}$ is non-negative, it is enough to show that $A_{ij}$ cannot be strictly positive. Suppose $A_{ij}>0$ for some $i>j$ with $i\in \mathcal{M}$ and $j\in \mathcal{M}^c$. Since $j\in \mathcal{M}^c$, there exists some $v\in \fE(w)$ with $v_k \geq \eta_k(w)$ for $k\neq j$ and $v_j > \eta_j(w)$. Define
 \[
  \tilde{v}_i:= \left\{\begin{array}{ll}-A_{ii} - \sqrt{A_{ii}^2 - 2\pare{\sum_{k=1}^{i-1} A_{ik} v_k + g_i(w)}}, & i\in \I \\ - A_{ii}^{-1}\pare{\sum_{k=1}^{i-1} A_{ik} v_k + g_i(w)}, & i\notin \I\end{array}\right..
 \]
 In either case, it is easy to see $\tilde{v}_i > \eta_i(w)$ for $i>j$ since $v_j > \eta_j(w)$. Also, the construction of $\tilde{v}_i$ gives $\tilde v_i \leq v_i$. This yields that the vector $\tilde{v} := (v_1, \ldots, v_{i-1}, \tilde v_i, v_{i+1}, \ldots, v_m)$ lies in $\fE(w)$ because $f(\tilde v, w) \leq f(v, w) \leq 0$. This is a contradiction to the assumption $i \in \mathcal{M}$.

By construction, $\eta(w)_{\mathcal{M}}$ is the unique equilibrium point of the subsystem $\dot y_i = (1/2) y_i^2\indic_{i \in \I} + \sum_{k \in \mathcal{M}}A_{ik} y_k + g_i(w)$, $i\in \mathcal{M}$, of \eqref{eq:riccati-V}. For any $y(0) \in \mathcal{S}_\infty(w)$, $\nu := \lim_t y(t)$ is an equilibrium point by Lemma~\ref{lem:longtermproperty}, and it satisfies $\nu_{\mathcal{M}} = \eta(w)_{\mathcal{M}}$. Therefore, $y(0)_{\mathcal{M}}$ belongs to the stable set of the aforementioned subsystem and the reverse inclusion is clear.
\end{proof}

We supplement Lemma \ref{lem: S-inf bou} with an additional comment. When $w \in \partial \fD$, every equilibrium point of \eqref{eq:riccati-V} is non-hyperbolic. In such cases, at least locally, the behavior of a solution trajectory is described by {\it stable}, {\it unstable manifolds}, and additionally, a \emph{center manifold}, which is determined by a certain partial differential equation. We refer the reader to \cite{Perko} for more details about this topic. For more details on theoretical and numerical studies of stable manifolds, see \cite{ChengM} or \cite{Osinga}. Next, we prove the decomposition of $\mathcal{S}_\infty$.

\begin{proof}[Proof of Theorem~\ref{thm: S-inf}]
 If $u = (v, w) \in \mathcal{S}_\infty$, then $w \in \fD$; otherwise, Corollay~\ref{cor:outsideD} implies that $y(t)$ with $y(0) = v$ blows up in finite time. Moreover, if $u \in \mathcal{S}_\infty^\circ$, then we claim that $v \in \mathcal{S}(w)$ with $w \in \fD^\circ$. To see this, first note that $ru \in \mathcal{S}_\infty$ for some $r> 1$ sufficiently close to $1$. Hence, $r w \in \fD$ which is equivalent to $\fE(rw) \neq \emptyset$. It then follows from Lemma \ref{lem: dense in fD} that $\fE^\circ(w) \neq \emptyset$. Therefore, $w \in \fD^\circ$ and $v \in \mathcal{S}_\infty^\circ(w) = \mathcal{S}(w)$. We then have $\mathcal{S}^\circ_\infty \subseteq \bigcup_{w\in \fD^\circ} \mathcal{S}(w) \times \set{w}$.

 Conversely, pick any $v \in \mathcal{S}(w)$ with $w \in \fD^\circ$. We want to show that, for any $u'=(v', w')$ sufficiently close to $u = (v,w)$, we have $u' \in \mathcal{S}_\infty$. Hence $\bigcup_{w\in \fD^\circ} \mathcal{S}(w) \times \{w\} \subseteq \mathcal{S}_\infty^\circ$. To this end, we consider, as in the proof of Lemma~\ref{lem: S-inf int}, a modified system $\dot r = (1/2)r_\I^{(2)} + \hat A r + g(w') - g(w)$ where $r(t) = y(t) - \eta(w)$, $y(0) = v'$, and $\hat A = A^\V + diag_{\mathcal{I}}(\eta(w))$. Recall that $\hat A$ has negative diagonals. This fact with straightforward calculations would imply that $r(t)$ stays bounded above as long as $|w'-w|$ is small and $r(0)$ is near the origin. Actually, it would be sufficient if we can find some $t_0$ such that $r(t_0)$ enters this neighborhood of the origin. However, $\lim_t r(t) = 0$ when $v' = v$ and the system for $r(t)$ is smooth. Hence, $r(t)$ continuously depends on its initial condition and thus such $t_0$ can be found.

 Finally, noticing that
 $
 \partial \mathcal{S}_\infty = \bigcup_{w\in\fD}\mathcal{S}_\infty(w) \times \set{w} \setminus \mathcal{S}_\infty^\circ,
 $
 we obtain the second statement from the first statement as well as from Lemmas~\ref{lem: S-inf int} and \ref{lem: S-inf bou}.
 \end{proof}

Now, we conclude this section with the proof of the characterization on exponential moments for affine processes.
\begin{proof}[Proof of Corollary~\ref{cor: exp moment X}]
 It follows from Theorem~\ref{thm:FM} that \eqref{eq: tran formula} is valid for any $X_0$  as long as $y(t)$ does not blow up by $T$. Then, the first statement is obtained. When $\expec\left[\exp\left(u^\top X_T\right)\right]<\infty$ for all $T\geq 0$,
 \[
 \begin{split}
  \frac1T \log \expec\bra{\exp(u^\top X_T)}
  = \frac12 (u^\D)^\top \pi_0 u^\D + \hat{b}^\D \cdot u^\D +\frac1T \int_0^T \hat{b}^\V \cdot y(s)\, ds + \frac1T y(T) \cdot X_0^\V + \frac1T  u^\D \cdot X_0^\D.
 \end{split}
 \]
 Lemma \ref{lem:longtermproperty} implies that $y(T)$ converges to an equilibrium point. Then \eqref{eq: exp moment growth rate} follows from sending $T\rightarrow \infty$ in the above identity.
\end{proof}

\section{Analysis of blow-up behavior}\label{sec:impvol}
In this section, we study solutions to \eqref{eq:riccati-V} which blow up in finite time. We first introduce a change of variables in Section \ref{subsec: blow up time}. Then we study the stability property of the system \eqref{eq:quadODE} and prove Lemma \ref{lem:cont blowup time} in Section \ref{subsec: blow up region}. Theorem \ref{thm:St-bdry} and Corollary \ref{cor2: exp moment X} are proved at last.

\subsection{Blow-up times}\label{subsec: blow up time}
To study the blow-up time $T^*(u)$, we employ a change of variables investigated by \cite{EliasG}:
\begin{equation}\label{eq:transf}
x(t) := \frac{2 y(t)}{1 + \sqrt{1 + 4|y(t)|^2}},
\end{equation}
where $y(t)$ is a solution to (\ref{eq:riccati-V}). This transform, which is equivalent to $y = x/(1 - |x|^2)$ with $|x| < 1$, maps $\mathbb{R}^m$ onto the open unit ball in $\mathbb{R}^m$. Moreover, $|x(t)|$ goes to the unit sphere whenever $|y(t)|$ goes to infinity. Therefore, this transform compactifies $\Real^m$. Using this transform, we have the following representation of the blow-up time.

\begin{prop} \label{prop:blowuptime}
For each fixed $u \in \mathbb{R}^m\times \operatorname{Ker} A^\D$, the blow-up time $T^*(u)$ is given by
\begin{equation}\label{eq: explosion time}
T^*(u) = \int_0^\infty \left( 1 - R(s)^4\right) ds,
\end{equation}
where $R(s)^2 := \sum_{i=1}^m x_i(s)^2$ and $x(s)$ solves
\begin{equation}\label{eq:tildef}
\frac{dx}{ds} = (1+R^2)\tilde f(x, w) - 2(x\cdot \tilde f(x, w))x,
\end{equation}
with $x(0) = 2v/(1 + \sqrt{1 + 4|v|^2})$ and $\tilde f(x, w) := (1-R^2)^2 f\left( x/(1-R^2), w\right)$.
\end{prop}
\begin{proof} The proof of \eqref{eq: explosion time} is essentially given in Section 3 of \cite{EliasG}. We present their argument here for the reader's convenience. First, it is straightforward to check that (\ref{eq:transf}) yields
\begin{equation}\label{eq:temptime}
\frac{dx}{dt} = \frac{1}{1-R(t)^2}\left( \tilde f(x, w) - \frac{2x\cdot \tilde f(x, w)}{1+R(t)^2}x\right), \quad
    x(0) = \frac{2v}{1 + \sqrt{1+4|v|^2}}.
\end{equation}
Second, define $s(t)$ via $ds/dt = (1-R(t)^4)^{-1}$ for $0\leq t< T^*(u)$ with $R(t)= \sqrt{\sum_{i=1}^m x_i(t)^2}$. Since $R(t)^2 < 1$ for $t<T^*(u)$, $s(t)$ is a strictly increasing function with the unique inverse $t(s)$. Let us write $x(t(s))$ and  $R(t(s))$ as $x(s)$ and $R(s)$, respectively. Then, (\ref{eq:tildef}) follows from changing the variable $t$ in (\ref{eq:temptime}) to $s$.

On the other hand, since
$$
\frac{d}{ds}(1 - R(s)^2) = - 2 x\cdot \frac{dx}{ds} = -2 (x\cdot \tilde f(x, w))(1 - R(s)^2),
$$
it follows that $1-R(s)^2 = (1-R(0)^2)\exp\left( - 2\int_0^s x(r)\cdot \tilde f(x(r), w)dr\right)$. Note that the integrand in this identity is uniformly bounded due to $|x(r)|\leq 1$, then $1-R(s)^2 > 0$ for any finite $s$. Thus, $s$ maps $[0, T^*(u))$ to $[0, \infty)$, hence \eqref{eq: explosion time} follows.
\end{proof}

The explosion time has the following property.
\begin{lem}\label{l:derivT}
For each fixed $u \in \mathbb{R}^m \times \operatorname{Ker} A^\mathcal{D}$, define $\mathcal{P}(u) = \set{p\in \Real_+ \such T^*(p u)<\infty}$. Then $T^*(\cdot u)$ is a nonincreasing and differentiable function on $\mathcal{P}(u)$.
\end{lem}
\begin{proof}
Note that $0\notin \mathcal{P}$ because $T^*(0) = \infty$. Let us now prove the nonincreasing property of $T^*(\cdot u)$. Without loss of generality, we assume $1\in \mathcal{P}(u)$. Then it suffices to prove $T^*(pu) \geq T^*(u)$ for any $p<1$.

Consider a solution $(y(t; p),z(t; p))$ to \eqref{eq:riccati} with initial condition $pu$ for $p \leq 1$. Then, $y(t; p)$  satisfies
$
\dot y_i = (1/2) y_i^2 \indic_{i \in \mathcal{I}} + (A^\V y)_i + (p^2/2)w^\top \pi_i w + p (A^\mathcal{C}w)_i
$
with $y(0; p) = p v$.  Now, define $\tilde y(t; p) = y(t; p)/p$ to obtain
$$
\dot{\tilde y} = \frac{p}{2}\tilde y_\mathcal{I}^{(2)} + A^\V \tilde y + \frac{p}{2}\left( \begin{array}{c}
            w^\top \pi_1 w \\ \vdots \\ w^\top \pi_m w \end{array}\right) + A^\mathcal{C}w \leq
            \frac{1}{2}\tilde y_\mathcal{I}^{(2)} + A^\V \tilde y + \left( \begin{array}{c}
            w^\top \pi_1 w \\ \vdots \\ w^\top \pi_m w \end{array}\right) + A^\mathcal{C}w, \quad \tilde y(0; p) = v.
$$
This implies that $\dot{\tilde y}(t; p) - f(\tilde y(t; p), w) \leq 0 = \dot y(t; 1) - f(y(t; 1), w)$. It then follows from the comparison theorem before Lemma~\ref{lem: comparison} that $y(t; p) = p \tilde{y} (t; p) \leq p y(t; 1)$ for all $0\leq t<T^*(u)$. This implies that $T^*(p u) \geq T^*(u)$.

For the differentiability of $T^*(\cdot u)$, first observe
\[
 \frac{d}{dp} R(s)^2 = \frac{4}{\sqrt{1+ 4|y(t(s); p)|^2} \left(1+ \sqrt{1+4|y(t(s); p)|^2}\right)^2} \times \left( \sum_{i=1}^m y_i(t(s);p)\frac{d}{dp}y_i(t(s);p)\right),
\]
where $t(s)$ is the function defined in Proposition~\ref{prop:blowuptime}. The analysis in the previous paragraph implies that $(d/dp)(y(t;p)/p) \geq 0$, which in turn gives $dy/dp \geq y/p$. On the other hand, Remark~\ref{rem: bdd below} implies that $\lim_{t\to T^*(pu)} y_i(t; p) = \infty$. Here $i$ is one component such that $y_i(t; p)$ explodes at $T^*(p u)$. Combining the previous observations, we obtain that
\[
 y_i(t(s); p) \frac{d}{dp} y_i(t(s); p) \geq \frac{1}{p} y_i^2(t(s); p),
\]
which is unbounded from above. Moreover, thanks to Lemma~\ref{lem:longtermproperty}, $y_i (d/dp)\, y_i$ dominates other non-explosive components. Therefore, there exists some sufficient large $s_0$ such that $(d/dp) R(s)^2>0$ for all $s\geq s_0$. Hence it follows from \eqref{eq: explosion time} that
\[
 \frac{d}{dp} T^*(p u) 
 = \frac{d}{dp} \left( \int_0^{s_0} \left(1 - R(s)^4\right)ds + \int_{s_0}^\infty \left(1 - R(s)^4\right) ds\right)
 = -\int_0^{\infty} 2 R(s)^2 \frac{d}{dp} R(s)^2 ds,
\]
where Fubini's and Tonelli's theorems are applied to integrals on $[0, s_0]$ and $[s_0, \infty)$, respectively.
\end{proof}

Not only does the change of variable \eqref{eq:transf} provide an expression for blow-up times, but also it helps to study the blow-up rate of solutions.

\begin{lem}\label{lem:firstblowup}
Suppose that $T^*(u) < \infty$ for some $u \in \mathbb{R}^m \times \operatorname{Ker} A^\mathcal{D}$, and that $l \leq m$ is the first component of $y(t)$ that blows up at $T^*(u)$. Then, $\lim_{t\uparrow T^*(u)} (T^*(u) - t) \,y_{l}(t) =c$ for some positive constant $c$.
\end{lem}
\begin{proof} For notational convenience, we write $T^*$ for $T^*(u)$.
First, we observe that $l$ must be in $\mathcal{I}$. Otherwise,
$
\dot y_l = A_{ll}y_l + \sum_{k=1}^{l-1}A_{lk}y_k + g_l(w)
$
where $y_1(t), \ldots, y_{l-1}(t)$ are finite on $[0, T^*]$, from which we infer that $y_l(t)$ must be finite in $[0, T^*]$. This contradicts with the choice of $l$.

Now, we apply the compactification \eqref{eq:transf} to $(y_1, \ldots, y_l)$. The resulting function $x(s)$, which is a vector valued function of length $l$, satisfies (\ref{eq:tildef}).
It then follows from the choice of $l$ and \eqref{eq:transf} that
$$
\lim_{s \rightarrow \infty} \frac{x(s)}{\sqrt{\sum_{i=1}^l x_i(s)^2}} = \lim_{t \uparrow T^*} \frac{y(t)}{\sqrt{\sum_{i=1}^l y_i(t)^2}} = e_l := \left\{ \begin{array}{l}
                                                                            0 \quad i < l\\
                                                                            1 \quad i = l
                                                                            \end{array} \right..
$$
The vector $e_l$ can be easily verified to be an equilibrium point of (\ref{eq:tildef}). If we denote the right hand side of (\ref{eq:tildef}) by $h(x) = (h_1, \ldots, h_l)$, then
$$
\frac{\partial h_i}{\partial x_j} = 2x_j\tilde f_i + (1 + R^2) \frac{\partial \tilde f_i}{\partial x_j} - 2\left(\sum_{k=1}^l x_k \tilde f_k \right)\delta_{ij} - 2x_i\left(\tilde f_j + \sum_{k=1}^l x_k \frac{\partial \tilde f_k}{\partial x_j}\right),
$$
where $\delta_{ij}$ is the Kronecker delta and
$$
\frac{\partial \tilde f_i}{\partial x_j} = x_i \indic_{i \in \mathcal{I}}\delta_{ij} - 2x_j \sum_{k=1}^i A_{ik}x_k + (1 - R^2)A_{ij} - 4(1 - R^2)x_j g(w).
$$
From above calculations, we obtain that the Jacobian matrix of $h$ at $e_l$ is $-I_{l}$. It is clear that the eigenvalues $\lambda_1, \ldots, \lambda_l$ of this Jacobian matrix are non-resonant, i.e., there is no $(m_1, \ldots, m_l)$ such that $m_k \in \{0\} \cup \mathbb{N}$, $\sum_{k=1}^l m_k \geq 2$, and $\lambda_k = \sum_{i=1}^l m_i \lambda_i$ for $k \in \{1, \ldots, l\}$. Theorem 4.1 in \cite{EliasG} now implies that
$
\sqrt{\sum_{k=1}^l y_k(t)^2} \sim c(T^* - t)^{-1}
$
as $t \uparrow T^*$ for some positive constant $c$. Since $y_l(t)$ is the first component in $y(t)$ that explodes at $T^*$, we consequently have
$
y_l(t) \sim c(T^* - t)^{-1}
$, as $t\uparrow T^*$.
\end{proof}

\subsection{Blow-up regions}\label{subsec: blow up region}
Before proving Theorem~\ref{thm:St-bdry} and Corollary~\ref{cor2: exp moment X} in this subsection, let us present a stability property of the system \eqref{eq:quadODE} and prove Lemma~\ref{lem:cont blowup time}.

\begin{lem}\label{lem:quadODE-blowup} If the trajectory of (\ref{eq:quadODE}) is bounded, then it converges to an equilibrium point. If the trajectory is unbounded, then it blows up in finite time.
\end{lem}
\begin{proof}
Since $x_{m+1}$ clearly converges to zero, which is the $(m+1)$-th coordinate of any equilibrium point of \eqref{eq:quadODE}, we only need to prove the first statement for the first $m$ coordinates. To this end, we prove by an induction on $i$. But we present the induction step only. The case $i=1$ is straightforward.

Now suppose that $x_1, \ldots, x_{i-1}$ converge to the first $i-1$ components of some equilibrium point. Recall that $x_i$ satisfies
$
 \dot x_i = \frac{T}{2}x_i^2\indic_{i \in \mathcal{I}} + \pare{T A_{ii} e^{-s} -1}x_i + T\sum_{k=1}^{i-1} A_{ik} e^{-s} x_k + T e^{-2s} g(w).
$
If $i\notin \I$, it then follows from the induction assumption and Lemma~\ref{lem:bdd-traj2} that $\lim_{s \rightarrow \infty} x_i(s) = 0$. If $i\in \I$, define ${\tilde x}(s) := T x_i(s) + T A_{ii} e^{-s} -1$, then $\tilde{x}$ satisfies
\[
 \dot {\tilde x} = \frac{1}{2}{\tilde x}^2 + T^2\sum_{k=1}^{i-1} A_{ik} e^{-s} x_k + T^2 e^{-2s} g(w) - \frac{1}{2}\left( TA_{ii}e^{-s}-1\right)^2 - TA_{ii}e^{-s}.
\]
Lemma \ref{lem: bounded traj} implies that ${\tilde x}(\infty) := \lim_{s\rightarrow \infty} {\tilde x}(s) \in \mathbb{R}$ and $(1/2){\tilde x}(\infty)^2 - 1/2 = 0$. Therefore, ${\tilde x}(\infty)= \pm 1$, which implies $x_i(\infty) = 0$ or $2/T$. Hence, in the above two cases, $x_i(s)$ converges to the $i$-th coordinate of some equilibrium point. This concludes the induction step.

For the second statement, let us assume that $x_i(s)$ is the first component whose trajectory is unbounded. It is necessarily that $i \in \mathcal{I}$; otherwise, we can utilize Lemma \ref{lem:bdd-traj2} to deduce that $x_i(s) \rightarrow 0$. Moreover, the trajectory of $x_i(s)$ is bounded from below. Indeed, using the comparison theorem for scalar ODEs, we can see that $x_i(s) \geq {\tilde x}(s)$ where ${\tilde x}$ solves
$$
\dot {\tilde x} = (T A_{ii}e^{-s} - 1){\tilde x} + T\sum_1^{i-1}A_{ik}e^{-s} x_k(s) + T e^{-2s} g_i(w), \quad {\tilde x}(0) = x_i(0).
$$
Again Lemma \ref{lem:bdd-traj2} implies that ${\tilde x}$ is bounded. Hence, the trajectory of $x_i$ is bounded from below.

Now, since $x_1, \ldots, x_{i-1}$ are bounded by assumption, we can find a positive constant $C$ such that
$$
\dot x_i \geq \frac{T}{2} x_i^2 + (T A_{ii} e^{-s} - 1)x_i -C =: h(x_i, s).
$$
Moreover, there exists a sufficiently large $r$ such that $r>T^{-1} \pare{1+ \sqrt{1+ 2T C}}$ and $h(r, s)>0$ for any $s\geq 0$. On the other hand, the trajectory of $x_i$ is unbounded from above by assumption, but bounded from below. Therefore, we can find some $s_0$ such that $x_i(s_0) =r$, and then $x_i(s)$ is strictly increasing from $s_0$ onward. As a result,
$$
\dot x_i \geq \frac{T}{2}x_i^2 - x_i - C, \quad s \geq s_0.
$$
However, notice that $x_i(s_0) > T^{-1} \pare{1+ \sqrt{1+ 2T C}}$ which is the unstable equilibrium of the previous ODE. Now it is immediate to check that $x_i(s)$ blows up in finite time. The second statement is proved.
\end{proof}

\begin{proof}[Proof of Lemma \ref{lem:cont blowup time}]
To show the continuity of $T^*(\cdot)$ on $\mathcal{P}$, let us denote the solution to \eqref{eq:riccati-V} by $y(t; u)$ where its dependence on $u$ is explicitly indicated. Let $T^*(u) = T$. The associated solution $x(s; u)$ to \eqref{eq:quadODE} does not blow up in finite time; otherwise, $y(t; u)$ explodes before $T$. Lemma~\ref{lem:quadODE-blowup} implies that the trajectory of $x(s; u)$ is bounded, which in turn yields that $\lim_{s \to \infty}x(s;u) = \nu'$ for some equilibrium point $\nu'$ of \eqref{eq:quadODE} and $\nu' = (\nu, 0)$ as set in the subsequent discussion.

However, this point has at least one nonzero coordinate. This is because, for the first index $l$ such that $y_l(t;u)$ blows up at $T$, we have from Lemma~\ref{lem:firstblowup} that $\lim_{t \uparrow T}(T-t)y_l(t;u) = c > 0$, which implies
\begin{equation}\label{eq: conv x_l}
\lim_{s \rightarrow \infty} x_l(s;u) = \lim_{t \uparrow T}x_l\left(-\log\left(1 - \frac{t}{T}\right); u\right) = \lim_{t \uparrow T} \left(1 - \frac{t}{T}\right)y_l(t;u) = \frac{c}{T}.
\end{equation}
Combined with the characterization of equilibria for \eqref{eq:quadODE}, it is necessarily that $c = 2$; see the discussion after \eqref{eq:quadODE}.
Moreover, $\nu_j = 0$ if $j \notin \I$. From this blow-up behavior of $y(t; u)$, it is easily deduced that there exists some $i \in \I$ such that $y_i/|y|$ converges to a positive constant while $\lim_t y_j/|y| = 0$ for all $j \notin \I$.

Now, let us consider $z(s; u)$ a solution to \eqref{eq:tildef} with $z(0; u) = 2v/(1 + \sqrt{1 + 4|v|^2})$. From $z/|z| = y/|y|$ and $\lim_s |z(s; u)| = 1$, $\lim_s z_i(s; u) > 0$ and $\lim_s z_j(s; u) = 0$ for all $j \notin \I$.
Therefore, for some positive $C$, small $\delta$ and sufficiently large $s_0$, $\sum_{k\in \I} z_k(s; u)^2 > C$ and $z_l(s; u) > - \delta$ for all $l=1, \ldots, m$ and $s\geq s_0$. Furthermore, since $z(s; u)$ continuously depends on $u$ (see e.g., \cite{Lefschetz}), there exists a sufficiently small neighborhood $U$ of $u$, such that $\sum_{k\in \I} z_k(s; u')^2 > C$ and $z_l(s; u') >-\delta$ for all $i=1, \ldots, m$, $s\geq s_0$, and all $u'\in U$.

Thanks to the analysis in the last paragraph, we can find a sufficiently small $\epsilon>0$ such that $\sum_{k\in \I} z_k^3(s; u')> \epsilon$ for all $s\geq s_0$ and $u'\in U$. As a consequence, for an even larger $s_0$, we can see that
\begin{eqnarray*}
\lefteqn{z(s; u')\cdot\tilde f(z(s; u'), w)} && \\
&=& \frac{1}{2}\sum_{k \in \I} z_k(s; u')^3 + \left(1 - R(s; u')^2\right) z(s; u')^\top A^\V z(s; u') + \left(1 - R(s;u')^2\right)^2 z(s; u')^\top g(w)
> \frac{\epsilon}{2}
\end{eqnarray*}
for all $s\geq s_0$ and $u'\in U$ because the second and third terms become small as $|z| \leq 1$, and $R^2$ converges to 1 as $s$ increases. Therefore, eventually we get
\[
1 - R(s; u')^2 \leq \left(1 - R(0)^2\right) \exp\left( - 2\int_0^{s_0} z(r; u')\cdot \tilde f(z(r; u'),w)dr  -\epsilon(s - s_0)\right),
\]
using the functional form of $R(s; u)^2$ in the proof of Proposition~\ref{prop:blowuptime}.
This facilitates the application of the dominated convergence theorem to conclude that $\lim_n T^*(u_n) = \int_0^\infty \lim_n \left(1 - R(s; u_n)^4\right) ds$ for a sequence of initial conditions $u_n \in U$ that converges to $u$. The right hand side of the previous identity, then, is equal to $T^*(u)$ due to the continuous dependence of $x(s; u)$ on $u$.
\end{proof}

\begin{proof}[Proof of Theorem~\ref{thm:St-bdry}] Let us consider the case of $\partial\mathcal{S}_T$ first.
The beginning paragraph of the proof of Lemma~\ref{lem:cont blowup time} argues that the solution $y(t)$ of \eqref{eq:riccati-V} that explodes at $T$ is associated with the function $x(s)$, the solution to \eqref{eq:quadODE} with $x(0) = (y(0), 1)$, and the limit $\lim_s x(s)$ is equal to some nonzero equilibrium point $\nu'=(\nu, 0)$ of \eqref{eq:quadODE}. In other words, $y(0) \in W^s_\nu(w, T)$. Hence, $\set{v \such T^*(v, w) = T} \subseteq \bigcup_{\nu \neq 0} W^s_{\nu}(w, T)$.

For the converse, suppose $v \in W^s_\nu (w, T)$ for some nonzero equilibrium point $\nu'=(\nu, 0)$. Recall the discussion following \eqref{eq:quadODE}, and thus $\nu_i$ is either zero or $2/T$. Since $\lim_{s\rightarrow \infty} x(s) = (\nu,0)$ when $x(s)$ is the solution to \eqref{eq:quadODE} with $x(0) = (v, 1)$, the same computation for $x_i(s)$ as in \eqref{eq: conv x_l} yields that $y(t)$ is finite for all $t < T$ and $\lim_{t\uparrow T} (T-t) y_i(t) =2$ for any $i$ such that $\nu_i \neq 0$. Therefore, $v \in \set{ v \such T^*(v,w) = T}$. This completes the proof of the first assertion. The second assertion then clearly follows.

As for $\mathcal{S}_T^\circ$, it is already noted in Section~\ref{subsec:blow-up} that $\lim_{s\uparrow \infty} x(s) = 0$ when $y(\cdot)$ blows up after $T$. Hence, $\mathcal{S}_T^\circ\subseteq \bigcup_{w\in \operatorname{Ker} A^\D} W_0^s(s, T) \times \{w\}$ is clear. The reverse inclusion also easily follows: If $x(s)$ converges to zero, then $T^*(u) \geq T$, but $T^*(u)=T$ cannot happen because, otherwise, the limit of $x(\cdot)$ would be a nonzero equilibrium point as shown above.
\end{proof}

\begin{proof}[Proof of Corollary~\ref{cor2: exp moment X}]
The first statement follows from the similar reasoning as in Corollary~\ref{cor: exp moment X}. Hence, we focus on \eqref{eq: blow-up rate}. From the transform formula \eqref{eq: tran formula}, for $S < T$ where $u \in \partial\mathcal{S}_T$,
$$
\log\mathbb{E}\left[\exp\left(u^\top X_S\right)\right] = \frac{S}{2}(u^\D)^\top \pi_0 u^\D + S\; \hat b^\D \cdot u^\D + \int_0^S \hat b^\V \cdot y(s)ds  + y(S)\cdot X^\V_0 + u^\D \cdot X^\D_0.
$$
Therefore, $\lim_{S\uparrow T} (T-S) \log \mathbb{E}[\exp(u^\top X_S)]= \lim_{S\uparrow T}(T-S)\int_0^S \hat b^\V \cdot y(s)ds
+ \lim_{S \uparrow T} (T-S)y(S)\cdot X^\V_0$. But, we know from Theorem~\ref{thm:St-bdry} and \eqref{eq: conv x_l} that $\lim_{S\uparrow T}(T-S)y(S) = T\nu$ for some $\nu$, the first $m$ components of some nonzero equilibrium point of \eqref{eq:quadODE}. Then, it is a simple exercise to show that $(T-S)\int_0^S \hat b^\V \cdot y(s)ds$ converges to zero. Now, the result is immediate.
\end{proof}

\section{Conclusion}\label{sec:conclusion}

In this paper, we study the long-term and blow-up behaviors of $\expec\left[\exp\left(u^\top X_T\right)\right]$ for multi-dimensional affine diffusion $X$ with some hierarchical structure between components. Analyzing solution behaviors of a multi-dimensional Riccati system, which is associated with a given affine diffusion process via the transform formula, we completely characterize sets of $u$ such that exponential moments are finite for all time or only before a fixed time. These sets are decomposed into the unions of stable sets for equilibrium points of the Riccati system or its transformed version. Then, we compute certain limits of exponential moments which provide detailed descriptions of behaviors of affine diffusions.

When the log-return of discounted stock prices is modeled by a linear transformation of affine diffusion processes, our results identify the long-term and blow-up behaviors of stock prices, especially in the case where the stock price moment is not explicitly known. These results provide a handle to investigate the implied volatility asymptotics for large-time-to-maturity, deep-out-of-money and deep-in-the-money options. We presented several examples to illustrate this point. Theoretically and practically, it still remains an interesting topic to extend the analysis of this paper to affine processes with jumps or affine processes on more general state spaces.

As a final remark, it is well known that in the literature of affine processes bond options and some other fixed income products can also be expressed in semi-closed form using the Fourier inversion formula. As long as the long-term growth rate of the underlying process satisfies the assumptions of the G\"artner-Ellis Theorem, we can calculate the asymptotic behaviors of the price of such a product, which are possibly useful in obtaining quantities that are analogues of the Black-Scholes implied volatility. We leave this as a potential future research topic.

\begin{appendix}

\section{Auxiliary results on ODEs and matrices}

\begin{lem}\label{lem: bounded traj}
 Let us consider a scalar ODE $\dot x(t) =  x(t)^2/2 + h(t)$ with $h(\infty):= \lim_{t\to \infty} h(t) \in \Real$. If the whole trajectory $\set{x(t) : t\geq 0}$ is bounded, then $h(\infty)\leq 0$, $x(\infty):= \lim_{t\to \infty} x(t)\in \Real$, and $x(\infty)^2/2 + h(\infty) =0$.
\end{lem}

\begin{proof}
 Let us prove $h(\infty)\leq 0$ first. Otherwise, $h(\infty) > 2\delta$ for some positive constant $\delta$. As a result, there exists some $t_0$ such that $x(t)^2/2 + h(t) > x(t)^2/2 + \delta$, for all $t\geq t_0$. It follows from the comparison theorem for scalar ODEs (see e.g. Chapter II of \cite{Hartman}) that $x(t)\geq y(t)$ for any $t\geq t_0$, where $y(t)$ is a solution to $\dot y= y^2/2 + \delta$ with $y(t_0)= x(t_0)$. However, a simple analysis of the previous ODE yields that $y(t)$ blows up in finite time. This contradicts to the assumption that the trajectory of $x(t)$ is bounded.

 To prove the rest of the statements, we shall first show that $\limsup_{t\to \infty} x(t)\leq \sqrt{-2 h(\infty)}$. If not, then there exists $\delta>0$ such that $\limsup_{t\to \infty} x(t)>\sqrt{-2 h(\infty) + 2\delta}$. Then, we can find $t_0$ such that $x(t_0) > \sqrt{-2 h(\infty) + 2\delta}$ and $h(t)\geq h(\infty) -\delta$ for all $t\geq t_0$. Next consider $y_-(t)$, which satisfies $\dot y_- = y_-^2/2 + h(\infty) -\delta$ with $y_-(t_0)= x(t_0)$. The comparison theorem implies that $x(t)\geq y_-(t)$ for $t\geq t_0$. However, $y_-$ explodes to infinity because $y(t_0) > \sqrt{-2 h(\infty) + 2\delta}$ and because the value on the right hand side is the unstable equilibrium point of the ODE satisfied by $y_-$. This contradicts to the boundedness assumption on $\set{x(t): t\geq 0}$.

 Now, if $\liminf_{t\to \infty} x(t) \geq \sqrt{-2h(\infty)}$, then, combined with the result from the last paragraph, we have $\lim_{t\to \infty} x(t)= \sqrt{-2h(\infty)}$ and we are done. To deal with the other case $\liminf_{t\to \infty} x(t) < \sqrt{-2h(\infty)}$, we separate $h(\infty)=0$ and $h(\infty)<0$ cases.

When $h(\infty)=0$, if $\liminf_{t\to \infty} x(t) < \sqrt{-2h(\infty)}$, then we have $\liminf_{t\to \infty} x(t) < \sqrt{2\delta}$ for any $\delta >0$. Thus, there exists $t_0$ such that $x(t_0)<\sqrt{2\delta}$ and $h(t)>-\delta$ for all $t\geq t_0$. Hence, $\lim_{t\to \infty} y_-(t)= -\sqrt{2\delta}$ because $y_-(t_0)= x(t_0)$ is less than the unstable equilibrium of $y_-$.
Combining this with $x(t)\geq y_-(t)$ for $t\geq t_0$, we obtain $\liminf_{t\to \infty} x(t) \geq \lim_{t\to \infty} y_-(t) = -\sqrt{2\delta}$, which implies $\liminf_{t\to \infty} x(t)\geq 0$ thanks to the arbitrary choice of $\delta$. This is a contradiction.

  When $h(\infty)<0$, if $\liminf_{t\to \infty} x(t) < \sqrt{-2h(\infty)}$, then there exists a sufficiently small positive $\delta$, such that $h(\infty)+\delta<0$ and $\liminf_{t\to \infty} x(t)<\sqrt{-2h(\infty) - 2\delta}$. We can find $t_0$ such that $x(t_0)< \sqrt{-2 h(\infty) - 2\delta}$ and $h(\infty)-\delta \leq h(t)\leq h(\infty)+ \delta$ for all $t\geq t_0$. Consider $y_+$ which satisfies $\dot y_+ = y_+^2/2 + h(\infty) + \delta$ and $y_+(t_0) = x(t_0)$. Note that $\sqrt{-2h(\infty) - 2\delta}$ is the unstable equilibrium of $y_+$, we then have from the comparison theorem that
  $
   -\sqrt{-2h(\infty) + 2\delta} = \lim_{t\to \infty} y_-(t) \leq \liminf_{t\to \infty} x(t) \leq \limsup_{t\to \infty} x(t) \leq \lim_{t\to \infty} y_+(t) = -\sqrt{-2 h(\infty) -2\delta}.
  $
  Since $\delta$ can be made arbitrarily small, we conclude from previous inequalities that $\lim_{t\to \infty} x(t) = -\sqrt{-2h(\infty)}$.
\end{proof}

\begin{lem}\label{lem:bdd-traj2} Let us consider a scalar ODE $\dot x(t) = (ae^{-t}-1)x(t) + g(t)$ with $g \in C^1$ and $g(\infty) := \lim_{t\rightarrow \infty} g(t) \in \mathbb{R}$. Then, $x(\infty) := \lim_{t \rightarrow \infty} x(t) \in \mathbb{R}$ and $x(\infty) = g(\infty)$.
\end{lem}
\begin{proof} Consider a new function $y(t) = \exp\left(ae^{-t}+t\right)x(t)$. Then, $y$ satisfies
$
\dot y(t) = \exp\left(a e^{-t} + t\right)g(t)
$
and consequently,
$$
x(t) = \exp\left( a - ae^{-t}-t\right)x(0) + \exp\left(-ae^{-t}-t\right) \int_0^t \exp\left(ae^{-s}+s\right)g(s)ds.
$$

Choose an arbitrary $\epsilon > 0$. Then, we can find a large $T = T(\epsilon) > 0$ such that $|g(t) - g(\infty)| \leq \epsilon$ for all $t \geq T$ and $\left| 1 - e^{ae^{-T}} \right| \leq \epsilon$. Next, we compute for $t > T$
\begin{eqnarray*}
\lefteqn{ e^{-t} \int_0^t e^{ae^{-s}+s}g(s)ds } && \\
&=& e^{-t} \int_0^T e^{ae^{-s}+s}g(s)ds + e^{-t}\int_T^t e^{ae^{-s}+s}g(\infty)ds  + e^{-t} \int_T^t e^{ae^{-s}+s}\left(g(s) - g(\infty)\right)ds.
\end{eqnarray*}
The last term is bounded by
$$
\left| e^{-t} \int_T^t e^{ae^{-s}+s}\left(g(s) - g(\infty)\right)ds \right| \leq \epsilon \max\left(e^{ae^{-T}},1\right)\left(1-e^{T-t}\right) \leq \epsilon(1+\epsilon)\left(1 - e^{T-t}\right),
$$
using $0 \leq e^{-s} \leq e^{-T}$. For the second term, we obtain
$
e^{ae^{-T}}\left(1 - e^{T-t}\right) \leq e^{-t}\int_T^t e^{ae^{-s}+s} ds \leq \left(1 - e^{T-t}\right)
$
if $a \leq 0$ (inequalities are reversed if $a \geq 0$). Therefore, by the assumption on $T$,
$$
\left| g(\infty) e^{-t}\int_T^t e^{ae^{-s}+s} ds - g(\infty)\left(1 - e^{T-t}\right)\right| \leq \epsilon |g(\infty)| \left( 1 - e^{T-t}\right).
$$
These calculations yield
$
\lim_{t \rightarrow \infty} e^{-t} \int_0^t e^{ae^{-s}+s}g(s)ds = g(\infty) + c
$
with $|c| \leq \epsilon\left( |g(\infty)| + 1+\epsilon\right)$.
Since $\epsilon$ is arbitrary, we can conclude that $\lim_{t\rightarrow \infty} x(t) = g(\infty)$.
\end{proof}

The following definition and property can be found in Chapter 6 of \cite{Berman-Plemmons}.
\begin{defn}\label{def: M-matrix}
 A square matrix is called a nonsingular $M$-matrix if it has non-positive off-diagonals and every real eigenvalue is positive.
\end{defn}
If $M$ is a nonsingular M-matrix, then $M^{-1}_{ij}\geq 0$ for all $i,j$.
Before we present the next result, we recall the following extension on the comparison theorem for scalar ODEs. A function $f:\Real^m \rightarrow \Real^m$ is \emph{quasi-monotone increasing}, if $f_k(x)\leq f_k(y)$ for any $x, y$ such that $x_k=y_k$ for some $k$ and $x_j\leq y_j$ for any $j\neq k$. Suppose that $f$ is quasi-monotone increasing and locally Lipschitz, then for any two differentiable functions $x, y: \Real_+ \rightarrow \Real^m$ with $x(a) \leq y(a)$,
\[
 \dot x(t)-f(x(t)) \leq \dot y(t)- f(y(t)), \quad \forall t\geq 0 \quad \implies \quad x(t)\leq y(t), \quad \forall t\geq 0.
\]
A proof can be found in \cite{Volkmann}.

\begin{lem}\label{lem: comparison}
 Let $\delta\in(0,1)$ and $M$ be a nonsingular M-matrix. Consider two solutions $x(t)$ and $y(t)$ of
 $
  \dot x = (1/2) \, x_\I^{(2)} - Mx
 $
 with $x(0)=\delta y(0)$ and $\I \subset \set{1, \ldots, m}$. Then $x(t)\leq \delta y(t)$ whenever they exist.
\end{lem}

\begin{proof}
 We define $f$ via $f(x):= (1/2)x_\I^{(2)} - Mx$. It is clearly locally Lipschitz as well as quasi-monotone increasing. Indeed, for any $x\leq y$ with $x_k=y_k$ for some $k$ and $x_j\leq y_j$ for all $j\neq k$, we have $f_k(x) = (1/2) x_k^2 \indic_{k\in \I} - \sum_{k=1}^m M_{kj} x_j \leq (1/2) x_k^2 \indic_{k\in \I} - M_{kk} x_k -\sum_{j\neq k} M_{kj} y_j = f_k(y)$, thanks to $M_{kj}\leq 0$ for $j\neq k$.

 Now consider $z(t) := x(t)/ \delta$ which satisfies
 $
\dot z = (\delta/2) z_\I^{(2)} - Mz
$ with $z(0) = y(0)$.
Then, we have
\[
\dot z(t) - f(z(t)) = \frac{\delta-1}{2} z_\I^{(2)} \leq 0 = \dot y(t) - f(y(t)).
\]
We conclude $z(t) \leq y(t)$ from the above comparison theorem.
\end{proof}

\end{appendix}

\bibliographystyle{plainnat}
\bibliography{biblio}

\begin{thebibliography}{32}
\providecommand{\natexlab}[1]{#1}
\providecommand{\url}[1]{\texttt{#1}}
\expandafter\ifx\csname urlstyle\endcsname\relax
  \providecommand{\doi}[1]{doi: #1}\else
  \providecommand{\doi}{doi: \begingroup \urlstyle{rm}\Url}\fi

\bibitem[Andersen and Piterbarg(2007)]{AndersenP}
L.~Andersen and V.~Piterbarg.
\newblock Moment explosions in stochastic volatility models.
\newblock \emph{Finance and Stochastics}, 11:\penalty0 29--50, 2007.

\bibitem[Benaim and Friz(2008)]{BenaimF08}
S.~Benaim and P.~Friz.
\newblock Smile asymptotics {II}: models with known moment generating
  functions.
\newblock \emph{Journal of Applied Probability}, 45:\penalty0 16--32, 2008.

\bibitem[Berman and Plemmons(1994)]{Berman-Plemmons}
A.~Berman and R.J. Plemmons.
\newblock \emph{Nonnegative matrices in the mathematical sciences}, volume~9 of
  \emph{Classics in Applied Mathematics}.
\newblock Society for Industrial and Applied Mathematics (SIAM), Philadelphia,
  PA, 1994.
\newblock Revised reprint of the 1979 original.

\bibitem[Calvet et~al.(2010)Calvet, Fisher, and Wu]{Wu10}
L.~E. Calvet, A.~J. Fisher, and L.~Wu.
\newblock Multifrequency cascade interest rate dynamics and dimension-invariant
  term structures.
\newblock 2010.
\newblock Working Paper.

\bibitem[Cheng et~al.(2004)Cheng, Ma, Lu, and Mei]{ChengM}
D.~Cheng, J.~Ma, Q.~Lu, and S.~Mei.
\newblock Quadratic form of stable sub-manifold for power systems.
\newblock \emph{International Journal of Robust and Nonlinear Control},
  14:\penalty0 773--788, 2004.

\bibitem[Chiang et~al.(1988)Chiang, Hirsch, and Wu]{Chiang88}
H.-D. Chiang, M.~W. Hirsch, and F.~F. Wu.
\newblock Stability regtions of nonlinear autonomous dynamical systems.
\newblock \emph{IEEE Transactions on Automatic Control}, 33:\penalty0 16--27,
  1988.

\bibitem[Cox et~al.(1985)Cox, Ingersoll, and Ross]{CIR}
J.~C. Cox, J.~E. Ingersoll, and S.~A. Ross.
\newblock A theory of the term structure of interest rates.
\newblock \emph{Econometrica}, 53:\penalty0 385--407, 1985.

\bibitem[Cuchiero et~al.(2011)Cuchiero, Filipovi\'c, Mayerhofer, and
  Teichmann]{CuchieroF}
C.~Cuchiero, D.~Filipovi\'c, E.~Mayerhofer, and J.~Teichmann.
\newblock Affine processes on positive semidefinite matrices.
\newblock \emph{Annals of Applied Probability}, 21:\penalty0 397--463, 2011.

\bibitem[Dai and Singleton(2000)]{DaiS}
Q.~Dai and K.~Singleton.
\newblock Specification analysis of affine term structure models.
\newblock \emph{Journal of Finance}, 55:\penalty0 1943--1978, 2000.

\bibitem[Duffie et~al.(2000)Duffie, Pan, and Singleton]{DPS}
D.~Duffie, J.~Pan, and K.~Singleton.
\newblock Transform analysis and asset pricing for affine jump-diffusions.
\newblock \emph{Econometrica}, 68:\penalty0 1343--1367, 2000.

\bibitem[Duffie et~al.(2003)Duffie, Filipovi\'c, and Schachermayer]{DFS}
D.~Duffie, D.~Filipovi\'c, and W.~Schachermayer.
\newblock Affine processes and applications in finance.
\newblock \emph{Annals of Applied Probability}, 13:\penalty0 984--1053, 2003.

\bibitem[Elias and Gingold(2006)]{EliasG}
U.~Elias and H.~Gingold.
\newblock Critical points at infinity and blow up of solutions of autonomous
  polynomial differential systems via compactification.
\newblock \emph{Journal of Mathematical Analysis and Applications},
  318:\penalty0 305--322, 2006.

\bibitem[Filipovi\'{c} and Mayerhofer(2009)]{Filipovic-Mayerhofer}
D.~Filipovi\'{c} and E.~Mayerhofer.
\newblock Affine diffusion processes: theory and applications.
\newblock \emph{Randon Series on Computational and Applied Mathematics},
  8:\penalty0 1--40, 2009.

\bibitem[Forde and Jacquier(2011)]{FordeJ2}
M.~Forde and A.~Jacquier.
\newblock The large-maturity smile for the heston model.
\newblock \emph{Finance and Stochastics}, 15:\penalty0 755--780, 2011.

\bibitem[Gatheral(2008)]{Gatheral}
J.~Gatheral.
\newblock Consistent modeling of {SPX} and {VIX} options.
\newblock 2008.
\newblock Slides. Available online at
  http://www.math.nyu.edu/fellow\_fin\_math/gatheral/gatheral.htm.

\bibitem[Glasserman and Kim(2010)]{GK10}
P.~Glasserman and K.~Kim.
\newblock Moment explosions and stationary distributions in affine diffusion
  models.
\newblock \emph{Mathematical Finance}, 20:\penalty0 1--33, 2010.

\bibitem[Goriely(2001)]{Goriely}
A.~Goriely.
\newblock Painlev\'e analysis and normal forms theory.
\newblock \emph{Physica D}, 152-153:\penalty0 124--144, 2001.

\bibitem[Hartman(1982)]{Hartman}
P.~Hartman.
\newblock \emph{Ordinary Differential Equations}.
\newblock Birkh\"{a}user, Boston, 1982.
\newblock 2nd edition.

\bibitem[Heston(1993)]{Heston}
S.~L. Heston.
\newblock A closed-form solution for options with stochastic volatility with
  applications to bond and currency options.
\newblock \emph{Review of Financial Studies}, 6:\penalty0 327--343, 1993.

\bibitem[Jacquier(2010)]{Jacquier}
A.~Jacquier.
\newblock \emph{Implied volatility asymptotics under affine stochastic
  volatility models}.
\newblock PhD thesis, Imperial College London, 2010.

\bibitem[Keller-Ressel(2011)]{Keller-Ressel}
M.~Keller-Ressel.
\newblock Moment explosions and long-term behavior of affine stochastic
  volatility models.
\newblock \emph{Mathematical Finance}, 21:\penalty0 73--98, 2011.

\bibitem[Keller-Ressel et~al.(2011)Keller-Ressel, Schachermayer, and
  Teichmann]{Keller-Resselb}
M.~Keller-Ressel, W.~Schachermayer, and J.~Teichmann.
\newblock Affine processes are regular.
\newblock \emph{Probability Theory and Related Fields}, 151:\penalty0 591--611,
  2011.

\bibitem[Kim(2010)]{Kim}
K.~Kim.
\newblock Stability analysis of {R}iccati differential equations related to
  affine diffusion processes.
\newblock \emph{Journal of Mathematical Analysis and Applications},
  364:\penalty0 18--31, 2010.

\bibitem[Lee(2004{\natexlab{a}})]{LeeRW}
R.~Lee.
\newblock Option pricing by transform methods: extensions, unification, and
  error control.
\newblock \emph{Journal of Computational Finance}, 7\penalty0 (3):\penalty0
  51--86, 2004{\natexlab{a}}.

\bibitem[Lee(2004{\natexlab{b}})]{RLee}
R.~Lee.
\newblock The moment formula for implied volatlity at extreme strikes.
\newblock \emph{Mathematical Finance}, 14:\penalty0 469--480,
  2004{\natexlab{b}}.

\bibitem[Lefschetz(1963)]{Lefschetz}
S.~Lefschetz.
\newblock \emph{Differential Equations: Geometric Theory}.
\newblock Interscience Publishers, New York, 1963.
\newblock 2nd edition.

\bibitem[Lewis(2000)]{Lewis}
A.~Lewis.
\newblock \emph{Option valuation under stochastic volatility}.
\newblock Finance Press, 2000.

\bibitem[Osinga et~al.(2004)Osinga, Rokni~Lamooki, and Townley]{Osinga}
H.~M. Osinga, G.~R. Rokni~Lamooki, and S.~Townley.
\newblock Numerical approximatoins of strong (un)stable manifolds.
\newblock \emph{Dynamical Systems}, 19:\penalty0 195--215, 2004.

\bibitem[Perko(2001)]{Perko}
L.~Perko.
\newblock \emph{Differential Equations and Dynamical Systems}.
\newblock Springer, New York, 2001.
\newblock 3rd edition.

\bibitem[Spreij and Veerman(2010)]{SpreijV}
P.~Spreij and E.~Veerman.
\newblock The affine transform formula for affine jump-diffusions with a
  general closed convex state space.
\newblock 2010.
\newblock Working Paper, University of Amsterdam.

\bibitem[Vasicek(1977)]{Vasicek}
O.~Vasicek.
\newblock An equilibrium characterization of the term structure.
\newblock \emph{Journal of Financial Economics}, 5:\penalty0 177--188, 1977.

\bibitem[Volkmann(1972)]{Volkmann}
P.~Volkmann.
\newblock Gew\"ohnliche differentialungleichungen mit quasimonoton wachsenden
  funktionen in topologischen vektorr\"aumen.
\newblock \emph{Mathematische Zeitschrift}, 127:\penalty0 157--164, 1972.

\end{thebibliography}
\end{document}